\documentclass[journal,onecolumn]{IEEEtran}
\usepackage{amsmath,amssymb}
\usepackage[font=small]{caption}
\usepackage{amsthm}
\usepackage{graphicx}
\usepackage{epsfig}
\usepackage{bm}
\usepackage{hyperref}
\usepackage{times}
\usepackage[usenames,dvipsnames]{color}
\usepackage{algorithm,algorithmic}

\newcommand{\oomit}[1]{}
\newtheorem{example}{Example}
\newtheorem{definition}{Definition}

\newtheorem{lemma}{Lemma}
\newtheorem{theorem}{Theorem}
\newtheorem{corollary}{Corollary}

\newtheorem{assumption}{Assumption}

\begin{document}

\title{Over- and Under-Approximating Reach Sets for Perturbed Delay Differential Equations\thanks{This work has been supported through grants by NSFC under grant No. 61872341, 61836005, 61625206 and 61732001, and by the CAS Pioneer Hundred Talents Program under grant No. Y8YC235015.}}

 \author{Bai Xue, Qiuye Wang, Shenghua Feng and Naijun Zhan
\thanks{State Key Lab. of Computer Science, Institute of Software, CAS, China and University of Chinese Academy of Sciences, China. Email:\{xuebai,wangqy,fengsh,znj\}@ios.ac.cn}
}
\maketitle

\IEEEpeerreviewmaketitle

\begin{abstract}
This note explores reach set computations for perturbed delay differential equations (DDEs). The perturbed DDEs of interest in this note is a class of DDEs whose dynamics are subject to perturbations, and their solutions feature the local homeomorphism property with respect to initial states. Membership in this class of perturbed DDEs is determined by conducting sensitivity analysis of solution mappings with respect to initial states to impose a bound constraint on the time-lag term. The homeomorphism property of solutions to such class of perturbed DDEs enables us to construct over- and under-approximations of reach sets by performing reachability analysis on just the boundaries of their permitted initial sets, thereby permitting an extension of reach set computation methods for ordinary differential equations to perturbed DDEs. Three examples demonstrate the performance of our approach. 
\end{abstract}

\section{Introduction}
\label{IN}
Reachability analysis, which involves computing appropriate approximations of reachable state sets, plays a fundamental role in computer-aided verification and
analysis \cite{Alur04}.  We have over the past decades witnessed a rapidly growing interest in developing reachability analysis techniques for dynamic systems modeled by ordinary differential equations (ODEs) or hybrid-state extensions thereof, e.g., \cite{Asarin01,Mitchell05,Girard08,Benvenuti08,Althoff13,Xin12,
xue2017just,xue2020} and the references therein. 

However, physical systems are often composed of networks of interacting systems, time delay phenomenon thereby exists ubiquitously and is appearing unavoidably. Delays are often involved in sensing or actuating by physical devices, in data forwarding to or from the controller, etc. Therefore, when conducting safety verification of such physical systems with time-delay phenomenon, DDEs are a suitable tool for modeling dynamics of these systems. The problem of performing reachability analysis for DDEs  is surely challenging. Growing attention is drawn to it recently \cite{Prajna05,Zou15,Huang16,feng2019taming}. Most of existing works, however, focused on over-approximating reach sets for systems modeled by DDEs with finite or infinite time horizon, not touching on the under-approximation problem of reach sets for DDEs. Recently, \cite{XueMFCLZ17} inferred a class of perturbation-free DDEs with solution mappings featuring the local homeomorphism property with respect to initial states. For such DDEs, the set-boundary reachability analysis method for ODEs in \cite{Xue16} and \cite{IEEE16} is extended to over- and under-approximate reach sets. A straightforward extension of the method in \cite{XueMFCLZ17} is to deal with DDEs with time-constant perturbations. \cite{Goubault18} extended the method in \cite{GoubaultP17} and the Taylor model based reachability method for ODEs to the computation of outer- and inner-approximations of reach sets for DDEs with time-constant perturbations. This work goes further than \cite{XueMFCLZ17} and \cite{Goubault18}, and studies the over- and under-approximate reachability analysis problem for DDEs subject to time-varying Lipschitz perturbations.

 Like \cite{XueMFCLZ17}, a constraint on the time-lag term is inferred based on the requirement that the sensitivity matrix is strictly diagonally dominant. The resulting DDE with time-lag terms satisfying this constraint has solutions featuring the local homeomorphism property. This constraint presented in this note is uniform over all perturbations ranging over a compact set. We topologically show that over- and under-approximations of certain reach sets can be computed by performing reachability analysis on just the  initial set's boundary. A computed over-approximation, which can be used to determine robust satisfiability of safety properties regardless of the actual perturbation, is a set which includes states reachable by all possible trajectories starting from legal initial states. A computed under-approximation, which can be used to determine robust violation of safety properties regardless of the actual perturbation, is a set of states in which for any perturbation each state is reachable by a trajectory initialized at some state in the initial set. Three illustrative examples demonstrate our approach.
 
The structure of this note is as follows. We formulate perturbed DDEs and the reachability problem of interest in Section \ref{Pre}. In Section \ref{RSC_section3} we present our reachability analysis approach. Before concluding this note in Section \ref{con}, we evaluate our approach on three examples in Section \ref{Ex}.
\section{Preliminaries}
\label{Pre}
The following notations are used throughout this note: the space of continuously differentiable functions on $\mathcal{X}$ is denoted by $\mathcal{C}^1(\mathcal{X})$; $\Delta^{\circ}$, $\Delta^c$ and $\partial \Delta$ represent the interior, complement and boundary of the set $\Delta$  respectively; vectors in $\mathbb{R}^n$ are denoted by boldface letters; $\|\bm{x}\|$ denotes the 2-norm, i.e., $\|\bm{x}\|=\sqrt{\sum_{i=1}^n x_i^2}$, where $\bm{x}=(x_1,\ldots,x_n)^T$; the set of $n \times n$ matrices over the field $\mathbb{R}$ of real numbers is denoted by $\mathbb{R}^{n\times n}$.

In this note we consider systems that can be modeled by DDEs of the following form
\begin{equation}
\label{dde5}
\dot{\bm{x}}(t)=\bm{f}(\bm{x}(t),\bm{x}_{\tau}(t),\bm{d}(t)), t \in [\tau,K\tau],
\end{equation}
where $\bm{x}(t)=(x_1(t),x_2(t),\ldots,x_n(t))^T:\mathbb{R}\rightarrow \mathbb{R}^n$, $\bm{x}_{\tau}(t)=(x_1(t-\tau),x_2(t-\tau),\ldots,x_n(t-\tau))^{T}: \mathbb{R}\rightarrow \mathbb{R}^n$, $\bm{d}(t)=(d_1(t),\ldots,d_m(t))^T: \mathbb{R}\rightarrow \mathcal{D}$ with $\mathcal{D}$ being a  compact set in  $\mathbb{R}^m$ is often used to incorporate model uncertainties and external disturbances, $K\geq 2$ is a positive integer and $\bm{f}\in \mathcal{C}^1(\mathbb{R}^n\times \mathbb{R}^n\times \mathcal{D})$. The initial condition $\bm{x}(\cdot): [0,\tau]\rightarrow \mathbb{R}^n$ for DDE \eqref{dde5} is governed by ODEs of the following form $$\dot{\bm{x}}(t)=\bm{g}(\bm{x}(t),\bm{d}(t)),t\in[0,\tau], \bm{x}(0) \in \mathcal{I}_0,$$ where 
$\bm{d}(t):\mathbb{R}\rightarrow \mathcal{D}, \bm{g}\in \mathcal{C}^1(\mathbb{R}^n\times \mathcal{D})$ and $\mathcal{I}_0\subset \mathbb{R}^n$ is a compact set. 

Denote the set of admissible perturbation inputs as
\[
\begin{split}
\widehat{\mathcal{D}}:&=\{\bm{d}(\cdot):[0,K\tau]\rightarrow \mathcal{D}~\textit{L-Lipschitz continuous}\}.
\end{split}
\] That is, there exists a uniform constant $L>0$ such that
\[\|\bm{d}(s_1)-\bm{d}(s_2)\|\leq L|s_1-s_2|, \forall s_1,s_2\in [0,K\tau], \forall \bm{d}\in \widehat{\mathcal{D}}.\] 
This assumption for perturbation inputs is relatively strong compared to what typically appears in the literature such as measurable inputs. However, the resulting set $\widehat{\mathcal{D}}$ is compact in the topology induced by  the uniform norm according to Arzel\`a-Ascoli Theorem \cite{rudin1964}. The compactness property of the set $\widehat{\mathcal{D}}$ is useful in proving the compactness of reach sets of interest in this note, which can be found in Lemma \ref{com}.

Based on the above assumption, we denote the trajectory of system \eqref{dde5} initialized at $\bm{x}_0\in \mathcal{I}_0$ and subject to a perturbation $\bm{d} \in \widehat{\mathcal{D}}$ by $\bm{\phi}(\cdot;\bm{x}_0,\bm{d}): [0,T']\rightarrow \mathbb{R}^n$, where $\bm{\phi}(0;\bm{x}_0,\bm{d})=\bm{x}_0$. Besides, we assume that $K\tau \leq T'$. In the following we define reach sets of the initial set $\mathcal{I}_0$ for $t\in [0,K\tau]$.
\begin{definition}
\label{re}
Given a perturbation $\bm{d}\in \widehat{\mathcal{D}}$, the reach set $\Omega(t;\mathcal{I}_0,\bm{d})$ at time $t\in [0,K\tau]$ is a set of states reached by trajectories starting from $\mathcal{I}_0$ after time duration $t$, i.e.,
$$\Omega(t;\mathcal{I}_0,\bm{d})=\{\bm{x}\in \mathbb{R}^n\mid \exists \bm{x}_0\in \mathcal{I}_0, \bm{x}=\bm{\phi}(t;\bm{x}_0,\bm{d})\}.$$
The reach set $\Omega_1(t;\mathcal{I}_0)$ 
at time $t\in [0,K\tau]$ is a set of states $\{\bm{x}\}$ visited by all trajectories originating from $\mathcal{I}_0$ after time duration $t$, i.e., $$\Omega_1(t;\mathcal{I}_0)=\{\bm{x}\in \mathbb{R}^n\mid  \exists \bm{d}\in \widehat{\mathcal{D}}, \exists \bm{x}_0\in \mathcal{I}_0, \bm{x}=\bm{\phi}(t;\bm{x}_0,\bm{d})\}.$$
The reach set $\Omega_2(t;\mathcal{I}_0)$ at time $t\in [0,K\tau]$ is a set of states $\{\bm{x}\}$ such that for every state $\bm{x}$ in it and every perturbation input $\bm{d}\in \widehat{\mathcal{D}}$, there exists a corresponding initial state $\bm{x}_0\in I$ such that $\bm{x}=\bm{\phi}(t;\bm{x}_0,\bm{d})$, i.e., $$\Omega_2(t;\mathcal{I}_0)=\{\bm{x}\in \mathbb{R}^n\mid \forall \bm{d}\in \widehat{\mathcal{D}}, \exists \bm{x}_0\in \mathcal{I}_0, \bm{x}=\bm{\phi}(t;\bm{x}_0,\bm{d})\}.\square$$
\end{definition}
The reach sets $\Omega_1(t;\mathcal{I}_0)$ and $\Omega_2(t;\mathcal{I}_0)$ are termed as the maximal and minimal forward reach sets in \cite{Mitchell07}, respectively. It is obvious that $\Omega_1(t;\mathcal{I}_0)=\cup_{\bm{d}\in \widehat{\mathcal{D}}}\Omega(t;\mathcal{I}_0,\bm{d})$ and $\Omega_2(t;\mathcal{I}_0)=\cap_{\bm{d}\in \widehat{\mathcal{D}}}\Omega(t;\mathcal{I}_0,\bm{d})$. In this note we focus on the computation of over-approximations of the reach set $\Omega_1(t;\mathcal{I}_0)$ and under-approximations of the reach set $\Omega_2(t;\mathcal{I}_0)$ for $t\in [0,K\tau]$. 
\begin{definition}
\label{OandU}
Given $t\in [0,K\tau]$, an over-approximation of the reach set $\Omega_1(t;\mathcal{I}_0)$ is a set $O(t;\mathcal{I}_0)$ satisfying $\Omega_1(t;\mathcal{I}_0)\subseteq O(t;\mathcal{I}_0).$ An under-approximation $U(t;\mathcal{I}_0)$ of the reach set $\Omega_2(t;\mathcal{I}_0)$ is a subset of $\Omega_2(t;\mathcal{I}_0)$, i.e., 
$U(t;\mathcal{I}_0)\subseteq \Omega_2(t;\mathcal{I}_0).$
\end{definition}

From Definition \ref{OandU}, an over-approximation $O(t;\mathcal{I}_0)$ is an enclosure such that $\bm{\phi}(t;\bm{x}_0,\bm{d})\in O(t;\mathcal{I}_0)$ for $\bm{x}_0 \in \mathcal{I}_0$ and $\bm{d}\in \widehat{\mathcal{D}}$, where $0\leq t\leq K\tau$. \textit{A useful property of such over-approximations is to determine robust satisfiability of safety properties.} If system \eqref{dde5} is modeled by an ODE, there are various existing methods for computing such over-approximations. For example, \cite{Althoff08} and \cite{Althoff13} proposed methods to perform over-approximate reachability analysis for ODEs subject to  Lipschitz continuous and piecewise continuous perturbations respectively. In contrast, an under-approximation $U(t;\mathcal{I}_0)$ is a set of states $\{\bm{x}\}$ such that for $\bm{x}\in U(t;\mathcal{I}_0)$ and $\bm{d}\in \widehat{\mathcal{D}}$ there exists a corresponding initial state $\bm{x}_0 \in \mathcal{I}_0$ satisfying $\bm{x}=\bm{\phi}(t;\bm{x}_0,\bm{d})$. \textit{A useful property of such under-approximations is to determine  robust violation of safety properties regardless of the actual perturbation.} \oomit{We believe the reach set $\Omega_2(t;\mathcal{I}_0)$ provides new tools towards robust property verification or control synthesis.} 

Throughout this note some additional assumptions are used.
\begin{assumption}
\label{assump}
1).  The viable evolution domain for system \eqref{dde5} is denoted by $\mathcal{X}$, a compact subset of $\mathbb{R}^n$. 2). The initial set $\mathcal{I}_0$ is a subset of $\mathcal{X}$. Also, the set of states visited by all trajectories starting from the initial set $\mathcal{I}_0$ within the time interval $[0,K\tau]$ is also included in $\mathcal{X}$, i.e., $\cup_{t\in [0,K\tau]}\Omega_1(t;\mathcal{I}_0)\subseteq \mathcal{X}$. (One technique to guarantee this assumption in our reachability computations will be demonstrated in Section \ref{Ex}.) 3). The infinity norms of matrices $\frac{\partial \bm{g}(\bm{x},\bm{d})}{\partial \bm{x}}$, $\frac{\partial \bm{f}(\bm{x},\bm{y},\bm{d})}{\partial \bm{x}}$ and $\frac{\partial \bm{f}(\bm{x},\bm{y},\bm{d})}{\partial \bm{y}}$ are uniformly bounded for $(\bm{x},\bm{y},\bm{d})\in \mathcal{X} \times \mathcal{X} \times \mathcal{D}$, i.e.,
\[\|\frac{\partial \bm{g}(\bm{x},\bm{d})}{\partial \bm{x}}\|_{\infty}\leq M',\|\frac{\partial \bm{f}(\bm{x},\bm{y},\bm{d})}{\partial \bm{x}}\|_{\infty}\leq M, \|\frac{\partial \bm{f}(\bm{x},\bm{y},\bm{d})}{\partial \bm{y}}\|_{\infty}\leq N,\]
where $M'$, $M$ and $N$ are positive real numbers. Since $g\in \mathcal{C}^1(\mathcal{X} \times \mathcal{D})$ and $\bm{f}\in \mathcal{C}^1(\mathcal{X} \times \mathcal{X} \times \mathcal{D})$, $M'$, $M$ and $N$ exist. 
\end{assumption}

\section{Reach Sets Computation}
\label{RSC_section3}
This section presents the set-boundary reachability method to compute over-approximations of the set $\Omega_1(t;\mathcal{I}_0)$ and under-approximations of the set $\Omega_2(t;\mathcal{I}_0)$ for a class of systems of the form \eqref{dde5}. This class of DDEs exhibits solution mappings featuring the local homeomorphism property with respect to initial states. We first derive a constraint on the time-lag term $\tau$ of DDE \eqref{dde5} such that the homeomorphism property is guaranteed. Then, we topologically prove that the boundaries of sets $\Omega_1(t;\mathcal{I}_0)$ and $\Omega_2(t;\mathcal{I}_0)$ for such class of systems can be retrieved by evolving the initial set's boundary. 

Theorem \ref{tau} presents the constraint on the time-lag term $\tau$ such that system \eqref{dde5} exhibits solutions featuring the homeomorphism property with respect to initial states.
\begin{theorem}
\label{tau}
If the time-lag term $\tau$ of DDE \eqref{dde5} satisfies
\begin{equation}
\label{uppp}
\tau \leq \min\left\{\frac{\epsilon-1}{\epsilon M'R},\frac{R-1}{M'R},\frac{\epsilon-1}{\epsilon R(M+N\epsilon)}, \frac{R-1}{R(M+N\epsilon)}\right\},
\end{equation}
where $R>1$ and $\epsilon>1$, then the solution mapping $\bm{\phi}(t;\cdot,\bm{d}):\mathcal{I}_0\rightarrow \Omega(t;\mathcal{I}_0,\bm{d})$ to system \eqref{dde5} is a homeomorphism between spaces $\mathcal{I}_0$ and $\Omega(t;\mathcal{I}_0,\bm{d})$ for $t\in [0,K\tau]$ and $\bm{d}\in \widehat{\mathcal{D}}$.
\end{theorem}

The derivation of constraint \eqref{uppp} is based on the requirement that the sensitivity matrix is strictly diagonally dominant as in \cite{XueMFCLZ17}. It is shown in Appendix. 
Comparing to Theorem 1 in \cite{XueMFCLZ17}, constraint \eqref{uppp} on the time-lag term $\tau$ gets rid of the explicit dependency on the dimension $n$, thereby avoiding possibly overly conservative requirement on $\tau$ such that the solution to system \eqref{dde5} exhibits the local homeomorphism property when the dimension $n$ of system \eqref{dde5} is too large. The underlying reason is that the derivation of constraint \eqref{uppp} only involves operations of the infinity norm of matrices. However, the derivation in \cite{XueMFCLZ17} involves manipulating 2-norm, infinity norm and max norm of matrices and their interconvertibility, thereby introducing the dimension $n$ into the estimate. 

If $\tau$ satisfies \eqref{uppp}, the solution mapping $\bm{\phi}(t;\cdot,\bm{d}):\mathcal{I}_0\rightarrow \Omega(t;\mathcal{I}_0,\bm{d})$ to system \eqref{dde5}, where $t\in [0,K\tau]$ and $\bm{d}\in \widehat{\mathcal{D}}$, maps the boundary and interior of the initial set $\mathcal{I}_0$ onto the boundary and interior of the set $\Omega(t;\mathcal{I}_0,\bm{d})$ respectively.

\subsection{Topological Analysis}
\label{RSC}
We in this subsection show that both reach sets $\Omega_1(t;\mathcal{I}_0)$ and $\Omega_2(t;\mathcal{I}_0)$ can be retrieved by performing reachability analysis on the initial set's boundary for system \eqref{dde5} with $\tau$ satisfying condition \eqref{uppp} in Theorem \ref{tau}. We firstly show that the solution $\bm{\phi}(t;\bm{x}_0,\bm{d})$ is continuous over $\bm{x}_0\in \mathcal{I}_0$ and $\bm{d}\in \widehat{\mathcal{D}}$.

\begin{lemma}
\label{conti}
If $\lim_{k \rightarrow \infty}\bm{d}_k(s)=\bm{d}(s)$  point-wise over $s\in [0,K\tau]$ and $\lim_{k\rightarrow \infty} \bm{x}_{k,0}=\bm{x}_0$, where $\bm{d}_k\in \widehat{\mathcal{D}}$ and $\bm{x}_{k,0}\in \mathcal{I}_0$,
then $\lim_{k\rightarrow \infty}\bm{\phi}(t;\bm{x}_{k,0},\bm{d}_k)= \bm{\phi}(t;\bm{x}_0,\bm{d})$ point-wise over $t\in [0,K\tau]$.
\end{lemma}
\begin{proof} 
As point-wise limits of $L$-Lipschitz continuous functions are $L$-Lipschitz continuous according to Proposition 1.2.4 in \cite{weaver1999}, $\bm{d}\in \widehat{\mathcal{D}}$ holds. Also, $\bm{x}_0\in \mathcal{I}_0$.

We first prove this statement:
Let $(\bm{h}_k)_k$ be a sequence of functions with $\|\bm{h}_k\|\leq B$, where $\bm{h}_k: [0,K\tau]\times \mathcal{X}\rightarrow \mathbb{R}^n$ is continuous over $t\in [0,K\tau]$ and uniformly $L'$-Lipschitz continuous over $\bm{x}\in \mathcal{X}$ for $t\in [0,K\tau]$. If $\lim_{k\rightarrow \infty}\bm{h}_k=\bm{h}$ point-wise, $\bm{x}(\cdot):[0,K\tau]\rightarrow \mathcal{X}$ and $\bm{x}_k(\cdot):[0,K\tau]\rightarrow \mathcal{X}$ are respectively solutions to $\dot{\bm{x}}(s)=\bm{h}(s,\bm{x}(s))$ \text{~and~} $\dot{\bm{x}}_k(s)=\bm{h}_k(s,\bm{x}_k(s))$,
 and $\bm{x}(\delta)=\lim_{k\rightarrow \infty} \bm{x}_k(\delta)$, then $\lim_{k\rightarrow \infty}\bm{x}_k(t)=\bm{x}(t)$ point-wise on $[\delta,K\tau]$ with $\delta\geq 0$. 

Since $\bm{x}(t)$ and $\bm{x}_k(t)$ satisfy 
$\bm{x}(t)=\bm{x}(\delta)+\int_{\delta}^{t} \bm{h}(s,\bm{x}(s))ds$ and 
$\bm{x}_k(t)=\bm{x}_k(\delta)+\int_{\delta}^t\bm{h}_k(s,\bm{x}_k(s))ds$, and $\|\bm{h}_k(s,\bm{x})-\bm{h}_k(s,\bm{y})\|\leq L'\|\bm{x}-\bm{y}\|$ for $\bm{x},\bm{y}\in \mathcal{X}$ and $s\in [0,K\tau]$, we have that
\begin{equation*}
\begin{split}
&\|\bm{x}(t)-\bm{x}_k(t)\|\leq \|\bm{x}(\delta)-\bm{x}_k(\delta)\|+\int_{\delta}^t \|\bm{h}(s,\bm{x}(s))-\bm{h}_k(s,\bm{x}(s))\|ds +\int_{\delta}^t\|\bm{h}_k(s,\bm{x}(s))-\bm{h}_k(s,\bm{x}_k(s))\| ds\\
&\leq \|\bm{x}(\delta)-\bm{x}_k(\delta)\|+\int_{\delta}^t \|\bm{h}(s,\bm{x}(s))-\bm{h}_k(s,\bm{x}(s))\|ds+\int_{\delta}^t L'\|\bm{x}(s)-\bm{x}_k(s)\| ds.
\end{split}
\end{equation*}
Since $\|\bm{x}(\delta)-\bm{x}_k(\delta)\|+\int_{\delta}^t \|\bm{h}(s,\bm{x}(s))-\bm{h}_k(s,\bm{x}(s))\|ds$ is non-decreasing with respect to $t$, Gr\"onwall's inequality implies
\[\|\bm{x}(t)-\bm{x}_k(t)\|\leq e^{L't-L'\delta}\|\bm{x}(\delta)-\bm{x}_k(\delta)\|+e^{L't-L'\delta}\int_{\delta}^t\|\bm{h}(s,\bm{x}(s))-\bm{h}_k(s,\bm{x}(s))\|ds.\]
Since $\|\bm{h}_k\| \leq B$ over $[0,K\tau]\times \mathcal{X}$ and $\lim_{k \rightarrow \infty}\|\bm{h}(s,\bm{x}(s))-\bm{h}_k(s,\bm{x}(s))\|=0$ for $s\in [0,K\tau]$, $\lim_{k \rightarrow \infty} \bm{x}_k(t)=\bm{x}(t)$ point-wise on $[\delta,K\tau]$ by dominated convergence and $\bm{x}(\delta)=\lim_{k\rightarrow \infty} \bm{x}_k(\delta)$.


We now prove the lemma by induction on $[0,i\tau]$, where $0\leq i\leq K$. Since $\bm{g}\in \mathcal{C}^1(\mathcal{X}\times \mathcal{D})$ and Assumption \ref{assump} implying that $\bm{\phi}(\cdot;\bm{x}_{k,0},\bm{d}_k):[0,K\tau]\rightarrow \mathcal{X}$ and $\bm{\phi}(\cdot;\bm{x}_0,\bm{d}):[0,K\tau]\rightarrow \mathcal{X}$, we obtain $\lim_{k \rightarrow \infty}\bm{\phi}(t;\bm{x}_{k,0},\bm{d}_k)=\bm{\phi}(t;\bm{x}_0,\bm{d})$ for $t\in [0,\tau]$ by setting $\bm{h}(s,\bm{x}):=\bm{g}(\bm{x},\bm{d}(s))$ and $\bm{h}_k(s,\bm{x}):=\bm{g}(\bm{x},\bm{d}_k(s))$. 

Let $\lim_{k \rightarrow \infty}\bm{\phi}(t;\bm{x}_{k,0},\bm{d}_k)=\bm{\phi}(t;\bm{x}_0,\bm{d})$ point-wise for $t\in [0, i\tau]$, where $0\leq i\leq K-1$, and $\lim_{k \rightarrow \infty} \bm{x}_{k,0}=\bm{x}_0$, we next show $\lim_{k \rightarrow \infty}\bm{\phi}(t;\bm{x}_{k,0},\bm{d}_k)=\bm{\phi}(t;\bm{x}_0,\bm{d})$ point-wise for $t\in [0, (i+1)\tau]$.

We have that $\lim_{k \rightarrow \infty}\bm{h}_k(s,\bm{x})=\bm{h}(s,\bm{x})$ point-wise over $s\in [0,(i+1)\tau]$, where $\bm{h}(s,\bm{x})=\bm{f}(\bm{x},\bm{\phi}(s-\tau;\bm{x}_0,\bm{d}),\bm{d}(s))$ \text{~and} $\bm{h}_k(s,\bm{x})=\bm{f}(\bm{x},\bm{\phi}(s-\tau;\bm{x}_{k,0},\bm{d}_k),\bm{d}_k(s)).$  Also, since $\bm{f}\in \mathcal{C}^1(\mathcal{X}\times \mathcal{X}\times \mathcal{D})$, we have that $\bm{h}_k$ and $\bm{h}$ satisfy the assumptions in the above statement. Let
\begin{equation}
\begin{split}
&\dot{\bm{x}}(s)=\bm{h}(s,\bm{x}(s)),  \bm{x}(i\tau)=\bm{\phi}(i\tau;\bm{x}_0,\bm{d}),\\
&\dot{\bm{x}}_k(s)=\bm{h}_k(s,\bm{x}(s)),  \bm{x}_k(i\tau)=\bm{\phi}(i\tau;\bm{x}_{k,0},\bm{d}_k).
\end{split}
\end{equation}
Since $\lim_{k \rightarrow \infty}\bm{x}_k(i\tau)=\bm{x}(i\tau)$, we have \[\lim_{k \rightarrow \infty}\bm{\phi}(t;\bm{x}_{0,k},\bm{d}_k)=\bm{\phi}(t;\bm{x}_0,\bm{d})\] point-wise over $t\in [i\tau, (i+1)\tau]$ and thus \[\lim_{k \rightarrow \infty}\bm{\phi}(t;\bm{x}_{0,k},\bm{d}_k)=\bm{\phi}(t;\bm{x}_0,\bm{d})\] point-wise over $t\in [0, (i+1)\tau]$. \hfill $\Box$
\end{proof}

Based on  Lemma \ref{conti}, we next show the compactness of reach sets $\Omega_1(t;\mathcal{I}_0)$ and $\Omega_2(t;\mathcal{I}_0)$, where $t\in [0,K\tau]$.
\begin{lemma}
\label{com}
If $\mathcal{I}_0$ is compact, both reach sets $\Omega_1(t;\mathcal{I}_0)$ and $\Omega_2(t;\mathcal{I}_0)$ are compact for $t\in [0,K\tau]$.
\end{lemma}
\begin{proof}
Since $\mathcal{I}_0\times \mathcal{D}$ is bounded, and $\bm{f}\in \mathcal{C}^1(\mathcal{X}\times \mathcal{X}\times \mathcal{D})$ and $\bm{g}\in \mathcal{C}^1(\mathcal{X}\times \mathcal{D})$, we have $\Omega_1(t;\mathcal{I}_0)$ is bounded. Next we prove that $\Omega_1(t;\mathcal{I}_0)$ is closed. Let $(\bm{x}_i)_{i=1}^{\infty}$ be a sequence with $\bm{x}_i\in \Omega_1(t;\mathcal{I}_0)$ and $\lim_{i\rightarrow \infty} \bm{x}_i=\bm{y}$. Then there exists a sequence $(\bm{x}_{0,i},\bm{d}_i)_{i=1}^{\infty}$ with $\bm{x}_{0,i}\in \mathcal{I}_0$ and $\bm{d}_i\in \widehat{\mathcal{D}}$ such that $\bm{x}_i=\bm{\phi}(t;\bm{x}_{0,i},\bm{d}_i)$. Since $(\bm{x}_{0,i},\bm{d}_i)\in \mathcal{I}_0\times \widehat{\mathcal{D}}$, according to Arzel\`a-Ascoli Theorem \cite{rudin1964}, we have that there exists a uniformly convergent sub-sequence $(\bm{x}_{0,i_k},\bm{d}_{i_k})_{k=1}^{\infty}$ in $(\bm{x}_{0,i},\bm{d}_i)_{i=1}^{\infty}$. Let $\bm{x}_0=\lim_{k\rightarrow \infty} \bm{x}_{0,i_k}$ and $\bm{d}_0(s)=\lim_{k\rightarrow \infty} \bm{d}_{i_k}(s)$ point-wise over $s\in [0,t]$. Obviously, $\bm{x}_0\in \mathcal{I}_0$. According to Proposition 1.2.4 in \cite{weaver1999}, $\bm{d}_0\in \widehat{\mathcal{D}}$ holds. Thus, from Lemma \ref{conti}, we have \[\bm{y}=\lim_{k\rightarrow \infty}\bm{\phi}(t;\bm{x}_{0,i_k},\bm{d}_{i_k})=\bm{\phi}(t;\bm{x}_{0},\bm{d}_{0}).\] Since $\bm{x}_0\in \mathcal{I}_0$ and $\bm{d}_0\in \widehat{\mathcal{D}}$, $\bm{y}\in \Omega_1(t;\mathcal{I}_0)$ and thus $\Omega_1(t;\mathcal{I}_0)$ is closed.

In the following we prove that $\Omega_2(t;\mathcal{I}_0)$ is compact. Since $\Omega_2(t;\mathcal{I}_0)\subseteq \Omega_1(t;\mathcal{I}_0)$ and $\Omega_1(t;\mathcal{I}_0)$ is compact, $\Omega_2(t;\mathcal{I}_0)$ is bounded. Therefore, we just show that $\Omega_2(t;\mathcal{I}_0)$ is closed.

Assume that there exists a sequence $(\bm{x}_i)_{i=1}^{\infty}$ with  $\bm{x}_i\in \Omega_2(t;\mathcal{I}_0)$ and $\lim_{i\rightarrow \infty}\bm{x}_i=\bm{y} \notin \Omega_2(t;\mathcal{I}_0)$. Without loss of generality, suppose that $\bm{y}$ can not be visited at time $t>0$ by any trajectory to system \eqref{dde5} subject to the perturbation $\bm{d}_1\in \widehat{\mathcal{D}}$ starting from the initial set $\mathcal{I}_0$ at time $0$.

Since $\bm{x}_i \in \Omega_2(t;\mathcal{I}_0)$, there exists a corresponding sequence $(\bm{x}_{i,0})_{i=1}^{\infty}$ such that $\bm{x}_i=\bm{\phi}(t;\bm{x}_{i,0},\bm{d}_1)$. Also, $\lim_{i\rightarrow \infty}\bm{\phi}(t;\bm{x}_{i,0},\bm{d}_1)=\bm{y}$. Since $\mathcal{I}_0$ is compact, there exists a convergent sub-sequence $(\bm{x}_{i_k,0})_{k=1}^{\infty}$ in $(\bm{x}_{i,0})_{i=1}^{\infty}$. Let $\bm{x}_0=\lim_{k\rightarrow \infty} \bm{x}_{i_k,0}$. Obviously, $\bm{x}_0\in \mathcal{I}_0$. According to Lemma \ref{conti}, we obtain $\bm{y}=\bm{\phi}(t;\bm{x}_{0},\bm{d}_1)$, contradicting the assumption that $\bm{y}$ is not visited at time $t>0$ by any trajectory of system \eqref{dde5} subject to the perturbation $\bm{d}_1\in \widehat{\mathcal{D}}$ starting from the initial set $\mathcal{I}_0$ at time $0$. Thus, $\bm{y}\in \Omega_2(t;\mathcal{I}_0)$, implying $\Omega_2(t;\mathcal{I}_0)$ is closed.
Therefore, $\Omega_2(t;\mathcal{I}_0)$ is compact. \hfill $\Box$
\end{proof}

Lastly, we present the core findings of this note, which are separately stated in Theorem \ref{b} and Theorem \ref{b1}.  Theorem \ref{b} shows that any trajectory starting from the boundary of the initial set $\mathcal{I}_0$ does not enter the interior of the reach set $\Omega_2(t;\mathcal{I}_0)$ for $t\in [0,K\tau]$. Theorem \ref{b1} formulates that each state in the boundary of the reach set $\Omega_1(t;\mathcal{I}_0)$ (or, $\Omega_2(t;\mathcal{I}_0)$) can be visited at time $t\in [0,K\tau]$ by a trajectory originating from the boundary of the initial set $\mathcal{I}_0$.

\begin{theorem}
\label{b}
Suppose that $\bm{\phi}(t;\cdot,\bm{d}):\mathcal{I}_0\rightarrow \Omega(t;\mathcal{I}_0,\bm{d})$ is a homeomorphism between spaces $\mathcal{I}_0$ and $\Omega(t;\mathcal{I}_0,\bm{d})$ for $\bm{d}\in \widehat{\mathcal{D}}$ and $t\in [0,K\tau]$. 
If $\mathcal{I}_0$ is compact and $\bm{x}_0\in \partial \mathcal{I}_0$, $\bm{\phi}(t;\bm{x}_0,\bm{d})\notin \Omega_2(t;\mathcal{I}_0)^{\circ}$ for $\bm{d}\in \widehat{\mathcal{D}}$.
\end{theorem}
\begin{proof}
Assume there exists $\bm{d}_1\in \widehat{\mathcal{D}}$ such that $\bm{\phi}(t;\bm{x}_0,\bm{d}_1)\in \Omega_2(t;\mathcal{I}_0)^{\circ}$. Then there exists $\delta>0$ such that $\{\bm{x}\in \mathbb{R}^n\mid \|\bm{x}-\bm{\phi}(t;\bm{x}_0,\bm{d}_1)\|\leq \delta\}\subset \Omega_2(t;\mathcal{I}_0).$

Since $\bm{\phi}(t;\cdot,\bm{d}_1):\mathcal{I}_0\rightarrow \Omega(t;\mathcal{I}_0,\bm{d}_1)$ is a homeomorphism between spaces $\mathcal{I}_0$ and $\Omega(t;\mathcal{I}_0,\bm{d}_1)$, we have $\bm{\phi}(t;\bm{x}_0,\bm{d}_1)\in \partial \Omega(t;\mathcal{I}_0,\bm{d}_1).$
Also, since $\Omega_2(t;\mathcal{I}_0)\subset \Omega(t;\mathcal{I}_0,\bm{d}_1)$, we obtain \[\{\bm{x}\in \mathbb{R}^n\mid \|\bm{\phi}(t;\bm{x}_0,\bm{d}_1)-\bm{x}\|\leq \delta\}\subset \Omega(t;\mathcal{I}_0,\bm{d}_1),\] implying 
\[\bm{\phi}(t;\bm{x}_0,\bm{d}_1)\in \Omega(t;\mathcal{I}_0,\bm{d}_1)^{\circ},\]  contradicting \[\bm{\phi}(t;\bm{x}_0,\bm{d}_1)\in \partial \Omega(t;\mathcal{I}_0,\bm{d}_1).\] Therefore, Theorem \ref{b} holds. \hfill $\Box$
\end{proof}

\begin{theorem}
\label{b1}
Suppose that $\mathcal{I}_0$ is compact and $\bm{\phi}(t;\cdot,\bm{d}):\mathcal{I}_0\rightarrow \Omega(t;\mathcal{I}_0,\bm{d})$ is a homeomorphism between spaces $\mathcal{I}_0$ and $\Omega(t;\mathcal{I}_0,\bm{d})$ for $t\in [0,K\tau]$ and $\bm{d}\in \widehat{\mathcal{D}}$. 1). If $\bm{x}\in \partial \Omega_1(t;\mathcal{I}_0)$, there exist $\bm{x}_0\in \mathcal{I}_0$ and $\bm{d}\in \widehat{\mathcal{D}}$ such that $\bm{x}=\bm{\phi}(t;\bm{x}_0,\bm{d})$. Moreover, $\bm{x}_0\in \partial \mathcal{I}_0$ holds. 2). If $\bm{x}\in \partial \Omega_2(t;\mathcal{I}_0)$, there exist $\bm{x}_0\in \partial \mathcal{I}_0$ and $\bm{d}\in \widehat{\mathcal{D}}$ such that $\bm{x}=\bm{\phi}(t;\bm{x}_0,\bm{d})$.
\end{theorem}
\begin{proof}
1) From Lemma \ref{com}, $\Omega_1(t;\mathcal{I}_0)$ is compact. Thus, there exist $\bm{x}_0\in \mathcal{I}_0$ and $\bm{d}\in \widehat{\mathcal{D}}$ such that $\bm{x}=\bm{\phi}(t;\bm{x}_0,\bm{d})$.

Assume that $\bm{x}\in \partial \Omega_1(t;\mathcal{I}_0)$ is visited at time $t>0$ by a trajectory to system \eqref{dde5} subject to the perturbation $\bm{d}\in \widehat{\mathcal{D}}$ originating from $\bm{x}_0\in \mathcal{I}_0^{\circ}$, i.e., $\bm{x}=\bm{\phi}(t;\bm{x}_0,\bm{d})$.

Since $\bm{\phi}(t;\cdot,\bm{d}):\mathcal{I}_0\rightarrow \Omega(t;\mathcal{I}_0,\bm{d})$ is a homeomorphism between spaces $\mathcal{I}_0$ and $\Omega(t;\mathcal{I}_0,\bm{d})$, we have $\bm{x}\in \Omega(t;\mathcal{I}_0,\bm{d})^{\circ}$. Also, since $\Omega(t;\mathcal{I}_0,\bm{d})\subset \Omega_1(t;\mathcal{I}_0)$, $\bm{x}\in \Omega_1(t;\mathcal{I}_0)^{\circ}$ holds, which contradicts $\bm{x}\in \partial \Omega_1(t;\mathcal{I}_0)$. Therefore, $\bm{x}_0\in  \partial \mathcal{I}_0$.

2) According to Lemma \ref{com}, $\Omega_2(t;\mathcal{I}_0)$ is compact. Thus, for any $\bm{d}\in \widehat{\mathcal{D}}$, there exists $\bm{x}^{\bm{d}}_0 \in \mathcal{I}_0$ such that $\bm{x}=\bm{\phi}(t;\bm{x}^{\bm{d}}_0,\bm{d})$. 

Assume that $\bm{x}^{\bm{d}}_0 \in \mathcal{I}_0^{\circ}$ for $\bm{d}\in \widehat{\mathcal{D}}$. Obviously, $\bm{x}\in \Omega_1(t;\mathcal{I}_0)^{\circ}$ since $\bm{x}\in \Omega(t;\mathcal{I}_0,\bm{d})^{\circ}$ and $\Omega(t;\mathcal{I}_0,\bm{d})\subseteq \Omega_1(t;\mathcal{I}_0)$. We denote the reverse map of $\bm{\phi}(t;\cdot,\bm{d}):\mathcal{I}_0\rightarrow \Omega(t;\mathcal{I}_0,\bm{d})$ by $\bm{\phi}^{-1}(t;\cdot,\bm{d}):\Omega(t;\mathcal{I}_0,\bm{d})\rightarrow \mathcal{I}_0$.
Therefore, $\bm{\phi}^{-1}(t;\bm{x},\bm{d})\in \mathcal{I}_0^{\circ}$ for $\bm{d}\in \widehat{\mathcal{D}}$.  It is obvious that \[\inf_{\bm{d}\in \widehat{\mathcal{D}}}\texttt{dist}(\bm{\phi}^{-1}(t;\bm{x},\bm{d}),\partial \mathcal{I}_0)\geq 0,\] where 
$\texttt{dist}(\bm{\phi}^{-1}(t;\bm{x},\bm{d}), \partial \mathcal{I}_0) = \inf_{\bm{z}\in \partial \mathcal{I}_0}\|\bm{\phi}^{-1}(t;\bm{x},\bm{d})-\bm{z}\|.$

Next, we show  that $\bm{\phi}^{-1}(t;\bm{x},\bm{d})$ is continuous over $\bm{d}\in \widehat{\mathcal{D}}$ and $\bm{x}\in \Omega_1(t;\mathcal{I}_0)$.  Let $\lim_{k \rightarrow \infty }\bm{d}_k(s)=\bm{d}(s)$ point-wise for $s\in [0,t]$ and $\lim_{k\rightarrow \infty}\bm{x}_k=\bm{x}$ with $\bm{d}_k\in \widehat{\mathcal{D}}$ and $\bm{x}_k\in \Omega_1(t;\mathcal{I}_0)$, $\bm{x}_0=\bm{\phi}^{-1}(t;\bm{x},\bm{d})$ and $\bm{x}_{0,k}=\bm{\phi}^{-1}(t;\bm{x}_k,\bm{d}_k)$, i.e., $\bm{x}=\bm{\phi}(t;\bm{x}_0,\bm{d})$ and $\bm{x}_k=\bm{\phi}(t;\bm{x}_{0,k},\bm{d}_k)$. We have that $\bm{x}_0\in \mathcal{I}_0$ and $\bm{x}_{0,k}\in \mathcal{I}_0$. We first prove that the sequence $(\bm{x}_{0,k})_{k=1}^{\infty}$ converges. If $(\bm{x}_{0,k})_{k=1}^{\infty}$ diverges, there exist two convergent sub-sequences $(\bm{x}_{0,k_{i}^1})_{i=1}^{\infty}$ and $(\bm{x}_{0,k_{i}^2})_{i=1}^{\infty}$ such that 
\begin{equation}
\label{conver}
\lim_{i\rightarrow \infty} \bm{x}_{0,k_{i}^1}\neq \lim_{i\rightarrow \infty} \bm{x}_{0,k_{i}^2}. 
\end{equation}
Let $\lim_{i\rightarrow \infty} \bm{x}_{0,k_{i}^1}=\bm{a}$ and $\lim_{i\rightarrow \infty} \bm{x}_{0,k_{i}^2}=\bm{b}$. Obviously, $\bm{a}\in \mathcal{I}_0$ and $\bm{b}\in \mathcal{I}_0$. However, Lemma \ref{conti} implies 
\[\bm{x}=\lim_{i \rightarrow \infty} \bm{\phi}(t;\bm{x}_{0,k_{i}^1},\bm{d}_{k_{i}^1})=\bm{\phi}(t;\bm{a},\bm{d})=\bm{\phi}(t;\bm{b},\bm{d})\] and thus $\bm{a}=\bm{b}$, contradicting \eqref{conver}. Thus, 
the sequence $(\bm{x}_{0,k})_{k=1}^{\infty}$ converges, and further Lemma \ref{conti} implies
\[\bm{x}=\lim_{k \rightarrow \infty} \bm{\phi}(t;\bm{x}_{0,k},\bm{d}_k)=\bm{\phi}(t;\lim_{k \rightarrow \infty} \bm{x}_{0,k},\lim_{k \rightarrow \infty} \bm{d}_k)=\bm{\phi}(t;\lim_{k \rightarrow \infty} \bm{x}_{0,k},\bm{d}),\]
implying $\bm{x}_0=\lim_{k \rightarrow \infty} \bm{x}_{0,k}$ and thus $\lim_{k \rightarrow \infty}\bm{\phi}^{-1}(t;\bm{x}_k,\bm{d}_k)=\bm{\phi}^{-1}(t;\bm{x},\bm{d})$. We below show that there exists $\bm{d}'\in \widehat{\mathcal{D}}$ such that \[\texttt{dist}(\bm{\phi}^{-1}(t;\bm{x},\bm{d}'),\partial \mathcal{I}_0)=\inf_{\bm{d}\in \widehat{\mathcal{D}}}\texttt{dist}(\bm{\phi}^{-1}(t;\bm{x},\bm{d}),\partial \mathcal{I}_0).\]
 Since $\widehat{\mathcal{D}}$ is compact, $\texttt{dist}(\cdot,\partial \mathcal{I}_0):\mathbb{R}^n \rightarrow \mathbb{R}$ is continuous and $\bm{\phi}^{-1}(t;\bm{x},\bm{d})$ is continuous over $\bm{d}\in \widehat{\mathcal{D}}$,  there exists $\bm{d}'\in \widehat{\mathcal{D}}$ such that \[\texttt{dist}(\bm{\phi}^{-1}(t;\bm{x},\bm{d}'),\partial \mathcal{I}_0)=\inf_{\bm{d}\in \widehat{\mathcal{D}}}\texttt{dist}(\bm{\phi}^{-1}(t;\bm{x},\bm{d}),\partial \mathcal{I}_0).\] The assumption that $\bm{\phi}^{-1}(t;\bm{x},\bm{d})\in \mathcal{I}_0^{\circ}$ for $\bm{d}\in \widehat{\mathcal{D}}$ implies $\bm{\phi}^{-1}(t;\bm{x},\bm{d}')\in \mathcal{I}_0^{\circ}$ and thus there exists $\epsilon_1>0$ such that $\texttt{dist}(\bm{\phi}^{-1}(t;\bm{x},\bm{d}'),\partial \mathcal{I}_0)\geq \epsilon_1.$ Thus, \[\texttt{dist}(\bm{\phi}^{-1}(t;\bm{x},\bm{d}),\partial \mathcal{I}_0)\geq \epsilon_1\] for $\bm{d}\in \widehat{\mathcal{D}}.$

Therefore, for every $\bm{d}\in \widehat{\mathcal{D}}$, there exists $\delta_{\bm{d}}>0$ such that \[\bm{\phi}^{-1}(t;\bm{y},\tilde{\bm{d}})\in \mathcal{I}_0\] for $\bm{y}\in B(\bm{x};\delta_{\bm{d}})$ and $\tilde{\bm{d}}\in B(\bm{d};\delta_{\bm{d}})$, where \[B(\bm{x};\delta_{\bm{d}})=\{\bm{y}\in \mathbb{R}^n\mid \|\bm{y}-\bm{x}\|\leq \delta_{\bm{d}}\}\] and \[B(\bm{d};\delta_{\bm{d}})=\{\tilde{\bm{d}}\in \widehat{\mathcal{D}}\mid \|\bm{d}(s)-\tilde{\bm{d}}(s)\|\leq \delta_{\bm{d}} \text{~for~} s \in [0,t]\}.\] Since $\widehat{\mathcal{D}}$ is compact with respect to uniform norm, there exists a finite sequence $(\delta_{\bm{d}_i})_{i=1}^N$ such that \[\widehat{\mathcal{D}}\subseteq \cup_{i=1}^N B(\bm{d}_i;\delta_{\bm{d}_i})\] and \[\bm{\phi}^{-1}(t;\bm{y},\tilde{\bm{d}})\in \mathcal{I}_0\] for $\bm{y}\in B(\bm{x};\delta_{\bm{d}_i})$ and $\tilde{\bm{d}}\in B(\bm{d}_i;\delta_{\bm{d}_i})$. Denoting $\delta:=\min_{1\leq i\leq N} \{\delta_{\bm{d}_i}\}$, we immediately have that
$\bm{\phi}^{-1}(t;\bm{y},\tilde{\bm{d}})\in \mathcal{I}_{0}$ for $\bm{y}\in B(\bm{x};\delta)$ and $\tilde{\bm{d}}\in \widehat{\mathcal{D}}$. Consequently, \[B(\bm{x};\delta)\subset \Omega_2(t,\mathcal{I}_0),\] which contradicts $\bm{x}\in \partial \Omega_2(t,\mathcal{I}_0)$. Therefore, we have that if $\bm{x}\in \partial \Omega_2(t;\mathcal{I}_0)$, there exist $\bm{x}_0\in \partial \mathcal{I}_0$ and $\bm{d}\in \widehat{\mathcal{D}}$ such that $\bm{x}=\bm{\phi}(t;\bm{x}_0,\bm{d})$. \hfill $\Box$
\end{proof}

From Theorem \ref{b1}, we conclude that the boundaries of both reach sets $\Omega_1(t;\mathcal{I}_0)$ and $\Omega_2(t;\mathcal{I}_0)$ are included in the reach set $\Omega_1(t;\partial \mathcal{I}_0)$ of the initial set's boundary, that is, $\partial \Omega_1(t;\mathcal{I}_0)\subseteq \Omega_1(t;\partial \mathcal{I}_0)$ and $\partial \Omega_2(t;\mathcal{I}_0)\subseteq \Omega_1(t;\partial \mathcal{I}_0)$.
According to Lemma 1 in \cite{Xue16}, Lemma \ref{com} and Theorem \ref{b1}, an over-approximation of $\Omega_1(t;\mathcal{I}_0)$ can be constructed by a set of the polytopic form, which includes $\Omega_1(t;\partial \mathcal{I}_0)$. Analogous to the algorithm in \cite{Xue16}, an under-approximation of $\Omega_2(t;\mathcal{I}_0)$ could be constructed by a set of the polytopic form, which excludes $\Omega_1(t;\partial \mathcal{I}_0)$. The latter statement is justified by Lemma \ref{under} below along with Lemma \ref{com} and Theorem \ref{b1}. Moreover, we observe from Theorem \ref{b} that the means of excluding $\Omega_1(t;\partial \mathcal{I}_0)$ to estimate an under-approximation of $\Omega_2(t; \mathcal{I}_0)$ does not induce extra conservativeness since no trajectory starting from $\partial \mathcal{I}_0$ at time $0$ enters the interior of the reach set $\Omega_2(t; \mathcal{I}_0)$ for $t\in [0,K\tau]$.
\begin{lemma}
\label{under}
Assume that $\mathcal{O}\subseteq \mathbb{R}^n$ is a compact set and $\mathcal{P}\subseteq \mathbb{R}^n$ is a compact convex polytope. If the boundary of $\mathcal{O}$ is a subset of the enclosure of the complement of $\mathcal{P}$, and the intersection of the interior of $\mathcal{O}$ and the interior of $\mathcal{P}$ is not empty, then $\mathcal{P}$ is an under-approximation of $\mathcal{O}$.
\end{lemma}

Lemma \ref{under} is a variant of Lemma 2 in \cite{Xue16}, in which $\mathcal{O}$ is required to be a simply connected compact set.  However, Lemma \ref{under} still can be assured by following the proof of Lemma 2 in \cite{Xue16}. Due to the space limitation, we omit its proof herein.

\subsection{Reach Sets Computation}
\label{rsc}
In this subsection we give a brief introduction on the set-boundary reachability method for over-approximating the set $\Omega_1(t;\mathcal{I}_0)$ and  under-approximating the set $\Omega_2(t;\mathcal{I}_0)$ for system \eqref{dde5} with $\tau$ satisfying condition \eqref{uppp} in Theorem 1, although one observes that this is a direct extension of the method in \cite{XueMFCLZ17}. Note that the set-boundary reachability method just guides existing reachability methods to perform computations on the initial set’s boundary. It helps to reduce computational
burden through reduction of volume of the initial set when performing reachability analysis, especially for cases with
large initial sets and/or large time horizons. Demonstrating its benefits is not the focus of this note since they were fully demonstrated in our previous works \cite{IEEE16,XueMFCLZ17}.

Assume that the initial set's boundary is represented as an union of $l$ subsets, that is, $\partial \mathcal{I}_0=\cup_{i=1}^l \mathcal{I}_{0,i}$. For $t \in [0,\tau]$, system \eqref{dde5} is governed by ODE $\dot{\bm{x}}(t)=\bm{g}(\bm{x}(t),\bm{d}(t))$. Therefore, we can apply reach set computation methods such as \cite{Althoff08} to the computation of an over-approximation $O(t;\partial \mathcal{I}_0)$ of the reach set of the initial set's boundary $\partial \mathcal{I}_0$ for $t \in [0,\tau]$, where $O(t;\partial \mathcal{I}_0)=\cup_{i=1}^l O(t;\mathcal{I}_{0,i})\cap \mathcal{X}$. According to Theorem \ref{b1}, the boundaries of reach sets $\Omega_1(t;\mathcal{I}_0)$ and $\Omega_2(t;\mathcal{I}_0)$ can be visited by trajectories of system \eqref{dde5} originating from the initial set's boundary, i.e., 
$\partial \Omega_1(t; \mathcal{I}_0 )\subseteq O(t;\partial \mathcal{I}_{0})$ and $\partial \Omega_2(t; \mathcal{I}_0 )\subseteq O(t;\partial \mathcal{I}_{0})$. According to Lemma 1 in \cite{Xue16} and Lemma \ref{under}, an over-approximation of the reach set $\Omega_1(t;\mathcal{I}_0)$ and under-approximation of the reach set $\Omega_2(t;\mathcal{I}_0)$ could be constructed by including and excluding the computed set $O(t;\partial \mathcal{I}_0)$ of the reach set of the initial set's boundary respectively. Based on above computations for the initial trajectory segment up to time $\tau$, the following steps are used to compute an over- and under-approximation of the reach sets $\Omega_1(t;\mathcal{I}_0)$ and $\Omega_2(t;\mathcal{I}_0)$ respectively for $t\in [k\tau,(k+1)\tau]$, $k=1,\ldots,K-1$.

1). We compute an over-approximation $O(t;\partial \mathcal{I}_0)$ of the reach set of the initial set's boundary at time $t$. This can be done in the following way: we compute an over-approximation $O(t;\mathcal{I}_{0,i})$ of the reach set $\Omega_1(t;\mathcal{I}_{0,i})$ for system (\ref{dde5}) with the initial set $O(k\tau;\mathcal{I}_{0,i})$, $i=1,\ldots,l$. This over-approximation can be computed by applying existing reachability analysis methods for ODEs subject to time-varying perturbation inputs $\bm{x}_{\tau}\in O(t-\tau;\mathcal{I}_{0,i})\cap  \mathcal{X}$ and $\bm{d}\in \widehat{\mathcal{D}}$. Therefore, \[O(t;\partial \mathcal{I}_{0})=\cup_{i=1}^m O(t;\mathcal{I}_{0,i})\cap \mathcal{X}\] is an over-approximation of the reach set of the initial set's boundary, i.e., $\Omega_1(t;\partial \mathcal{I}_0)\subseteq O(t;\partial \mathcal{I}_{0})$. According to Theorem \ref{b1}, the boundaries of reach sets $\Omega_1(t;\mathcal{I}_0)$ and $\Omega_2(t;\mathcal{I}_0)$ can be visited by trajectories of system \eqref{dde5} originating from the initial set's boundary. Consequently, $\partial \Omega_1(t; \mathcal{I}_0 )\subseteq O(t;\partial \mathcal{I}_{0})$ and $\partial \Omega_2(t; \mathcal{I}_0 )\subseteq O(t;\partial \mathcal{I}_{0})$.
 
The following two steps compute an over-approximation of the reach set $\Omega_1(t;\mathcal{I}_0)$ and an under-approximation of the reach set $\Omega_2(t;\mathcal{I}_0)$ by including and excluding the obtained over-approximation $O(t;\partial \mathcal{I}_0)$ respectively.

2). We compute a compact polytope $\mathcal{O}_{k,t}$ such that it covers the computed over-approximation $O(t;\partial \mathcal{I}_{0})$. The set $\mathcal{O}_{k,t}$ is an over-approximation of the reach set $\Omega_1(t;\partial \mathcal{I}_0)$ according to Lemma 1 in \cite{Xue16}.

3). We compute a compact polytope $\mathcal{U}_{k,t}$ satisfying these two conditions: 3a). The computed over-approximation $O(t;\partial \mathcal{I}_0)$ is a subset of the enclosure of its complement. Based on linear programs as in \cite{Xue16}, such $\mathcal{U}_{k,t}$ can be computed by shrinking the over-approximation $\mathcal{O}_{k,t}$ obtained in step 2) such that the enclosure of its complement includes the over-approximation $O(t;\partial \mathcal{I}_0)$.
3b). Its interior intersects the interior of reach set $\Omega(t;\mathcal{I}_0,\bm{d})$ for $\bm{d}\in \widehat{\mathcal{D}}$. For such sake, we first extract $\bm{x}_0 \in \mathcal{I}_0^{\circ}$, then compute an over-approximation $O(t;\bm{x}_0)$ of the reach set $\Omega_1(t;\bm{x}_0)$ of the state $\bm{x}_0$ at time $t$. Finally, we verify whether $O(t;\bm{x}_0)$ is included in $\mathcal{U}_{k,t}^{\circ}$. If $O(t;\bm{x}_0)$ is a subset of the interior of the set $\mathcal{U}_{k,t}^{\circ}$, we obtain that \[\mathcal{U}_{k,t}^{\circ}\cap \Omega(t;\mathcal{I}_0,\bm{d})^{\circ}\neq \emptyset\] for $\bm{d}\in \widehat{\mathcal{D}}$ and consequently \[\mathcal{U}_{k,t}\subseteq \Omega(t;\mathcal{I}_0,\bm{d})\] for $\bm{d}\in \widehat{\mathcal{D}}$ according to Lemma \ref{under}. Thus,\[\mathcal{U}_{k,t}\subseteq  \cap_{\bm{d}\in \widehat{\mathcal{D}}} \Omega(t;\mathcal{I}_0,\bm{d})=\Omega_2(t;\mathcal{I}_0)\] and consequently $\mathcal{U}_{k,t}$ is an under-approximation of the reach set $\Omega_2(t;\mathcal{I}_0)$.

\section{Examples and Discussions}
\label{Ex}
In this section we evaluate our set-boundary reachability analysis method on three examples of two two-dimensional systems and one seven-dimensional system. All computations were carried out on an i7-7500U 2.70GHz CPU with 32GB RAM running Windows 10. 

\begin{example}
\label{11}
Consider a simple two-dimensional system of the form \eqref{dde5}, where
\[\bm{g}(\bm{x},\bm{d})=(g_1(\bm{x},\bm{d}),g_2(\bm{x},\bm{d}))^T=(-0.1y+dx,-0.01x+0.02y)^T,\]
\[
\bm{f}(\bm{x},\bm{x}_{\tau},\bm{d})=(f_1(\bm{x},\bm{x}_{\tau},\bm{d}),f_2(\bm{x},\bm{x}_{\tau},\bm{d}))^T=(-0.1y+dx,-0.01x_{\tau}+0.02y)^T,\]
$\mathcal{D}=[-0.01,0.01]$, $\mathcal{X}=[-100,100]\times [-100,100]$ and $\mathcal{I}_0=[0.1,0.3]\times [0.1,0.3]$ with $\partial \mathcal{I}_0=\cup_{i=1}^4 \mathcal{I}_{0,i}$, where $\mathcal{I}_{0,1}=[0.1,0.1]\times [0.1,0.3]$, 
$\mathcal{I}_{0,2}=[0.3,0.3]\times [0.1,0.3]$, $\mathcal{I}_{0,3}=[0.1,0.3]\times [0.1,0.1]$ and $\mathcal{I}_{0,4}=[0.1,0.3]\times [0.3,0.3]$. 
\end{example}
In this example $M'=0.11$, $M=0.11$, $N=0.01$, $R=2$ and $\epsilon=4$. Actually, the presence of $\mathcal{X}$ is not necessary since $M'=0.11$, $M=0.11$, $N=0.01$, $R=2$ and $\epsilon=2$ hold true over $\bm{x}\in \mathbb{R}^n$ and $\bm{x}_{\tau}\in \mathbb{R}^n$. Through simple calculations, $\tau\leq 2.50$ satisfies Theorem \ref{tau}. Note that if we use the condition in \cite{XueMFCLZ17} to estimate $\tau$, via simple calculations with $\|\frac{\partial \bm{g}(\bm{x},\bm{d})}{\partial \bm{x}}\|_{\max}=0.1$, $\|\frac{\partial \bm{f}(\bm{x},\bm{x}_{\tau},\bm{d})}{\partial \bm{x}}\|_{\max}=0.1$, $\| \frac{\partial\bm{f}(\bm{x},\bm{x}_{\tau},\bm{d})}{\partial \bm{x}_{\tau}}\|_{\max}=0.01$, $R=2$ and $\epsilon=2$, we have $\tau\leq 0.66$.

Let $\tau=1$ and $K=10$, we perform reachability analysis for this example. The computed over-approximation $\cup_{t\in [0,10.0]} O(t;\partial \mathcal{I}_0)$ of the reach set of the initial set's boundary is illustrated in Fig. \ref{fig:10}, which also shows the computed over- and under-approximations at time $t=10$ as well as the corresponding over-approximation of the reach set of the initial set's boundary. The computation time is 15.53 seconds.

\begin{figure}[!htb]
\centering
  \includegraphics[width=.45\linewidth]{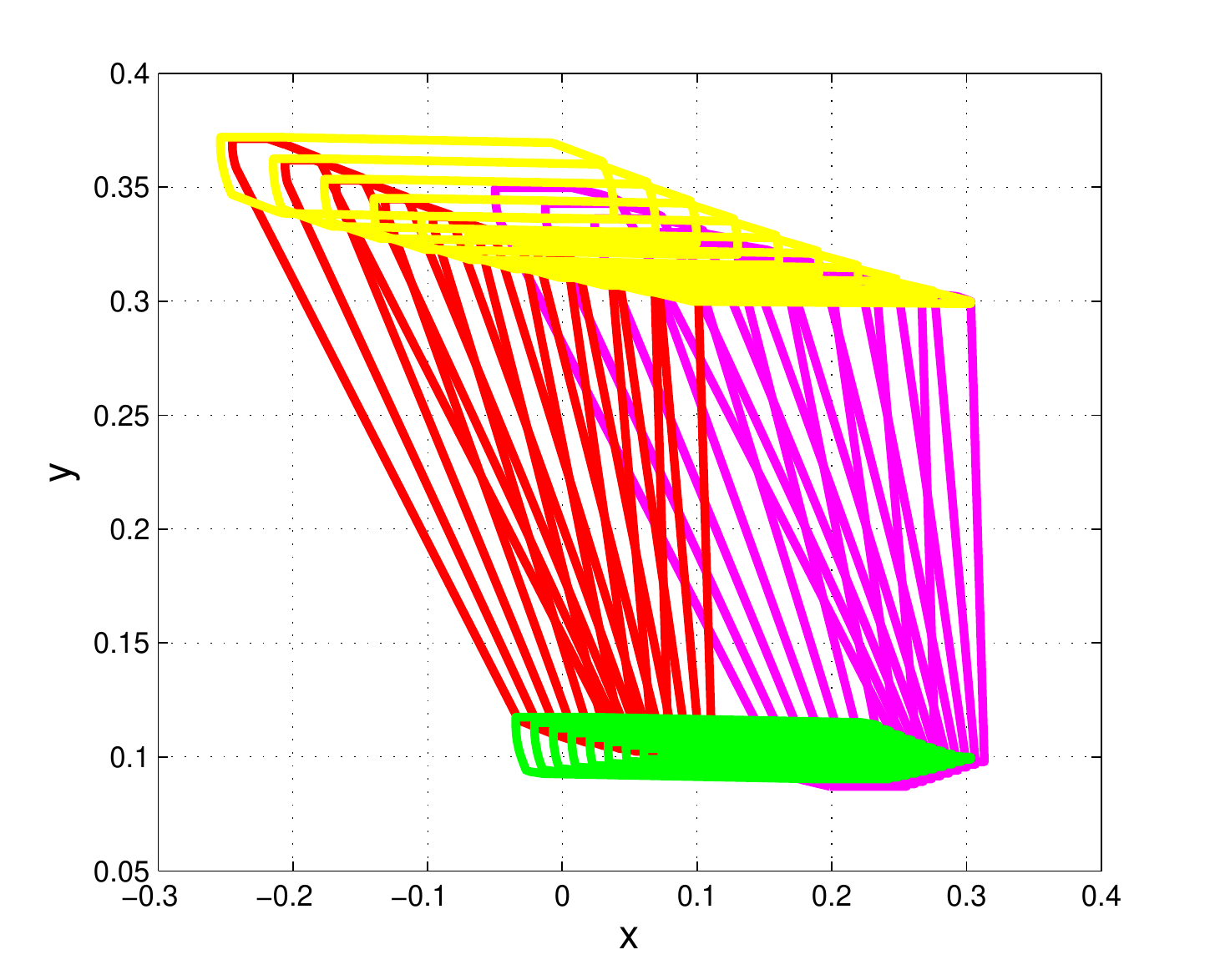}
  \hfill
 \includegraphics[width=.45\linewidth]{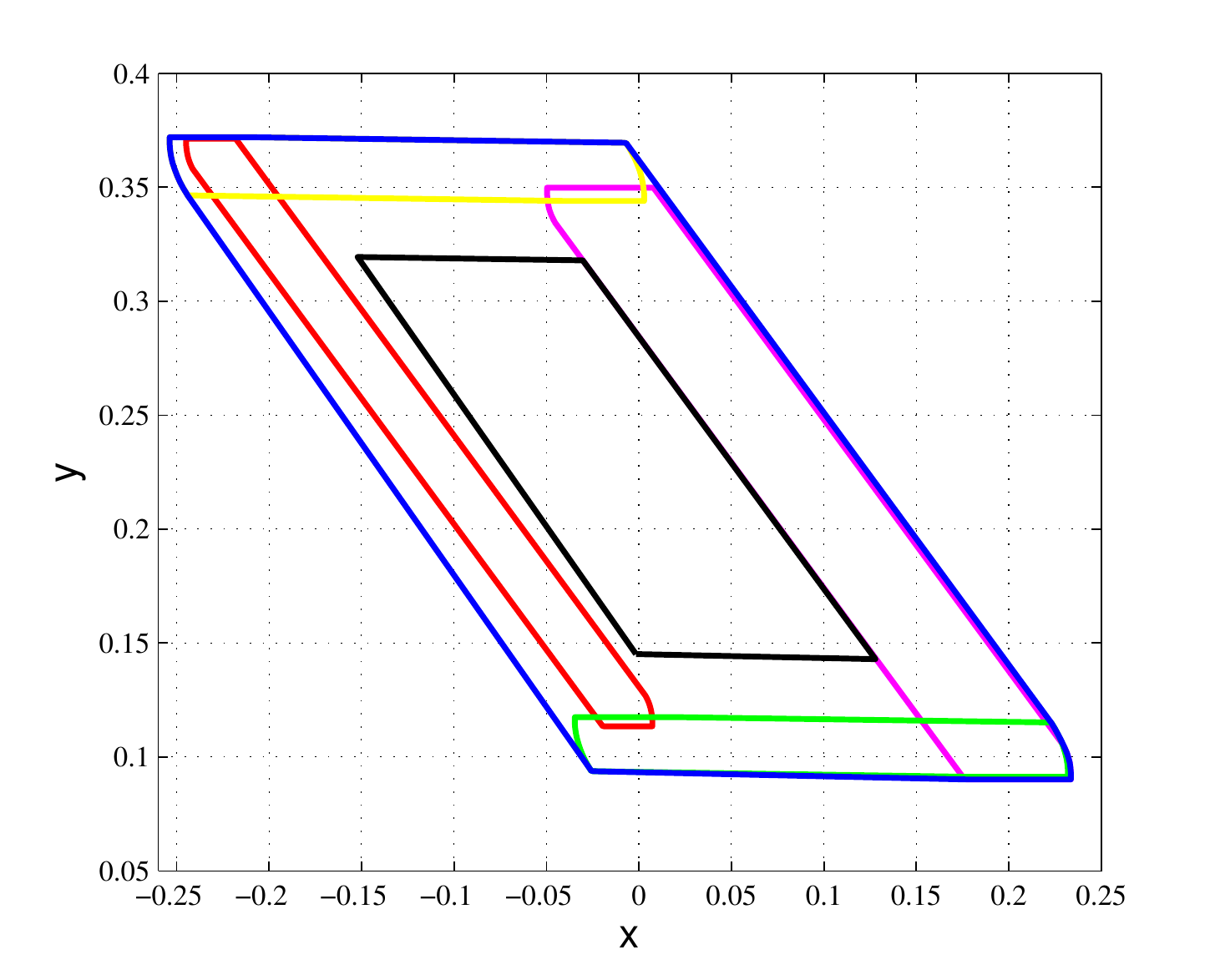}
\caption{Reach sets for Example \ref{11}. Left: Red, purple, green and yellow curves respectively denote the boundary of $\cup_{t\in [0,10.0]}O(t;\mathcal{I}_{0,1})$, $\cup_{t\in [0,10.0]} O(t;\mathcal{I}_{0,2})$, $\cup_{t\in [0,10.0]} O(t;\mathcal{I}_{0,3})$ and $\cup_{t\in [0,10.0]} O(t;\mathcal{I}_{0,4})$. Right: Red, purple, green and yellow curves respectively denote $\partial O(10;\mathcal{I}_{0,1})$, $\partial O(10;\mathcal{I}_{0,2})$, $\partial O(10;\mathcal{I}_{0,3})$ and $\partial O(10;\mathcal{I}_{0,4})$. Blue and black curves denote $\partial O(10;\mathcal{I}_0)$ and $\partial U(10;\mathcal{I}_0)$ respectively.}
\label{fig:10}
\end{figure}
Based on this example we motivate our study of over- and under-approximations of interest in this paper. Suppose the dynamics of a physical system are captured by this perturbed DDE. Unfortunately, this perturbation input is not known exactly. We wish to verify whether the physical system does not enter a set of unsafe states $\mathcal{X}_{u}$ at time $t=10$. If $\mathcal{X}_{u}=[0.15,0.2]\times [0.3,0.35]$, we conclude that the physical system starting from the initial set $\mathcal{I}_0$ is robustly safe since the computed over-approximation $O(10;\mathcal{I}_0)$ does not intersect $\mathcal{X}_{u}$, which can be observed from Fig. \ref{fig:10}. In contrast, if $\mathcal{X}_{u}=[0,0.05]\times [0.25,0.3]$, we observe from Fig. \ref{fig:10} that the computed under-approximation intersects $\mathcal{X}_{u}$. This indicates that the physical system starting from the initial set $\mathcal{I}_0$ will touch $\mathcal{X}_{u}$ at time $t=10$ regardless of the actual perturbation and consequently it is robustly unsafe.

\begin{example}
\label{lotka-volterra}
Consider a system, which is adapted from \cite{XueMFCLZ17}, of the form \eqref{dde5}, where 
\[\bm{g}(\bm{x},\bm{d})=(g_1(\bm{x},\bm{d}),g_2(\bm{x},\bm{d}))^T
=(y,-0.2x+2.0y-0.2x^2y+d)^T,\]
\[\bm{f}(\bm{x},\bm{x}_{\tau},\bm{d})=(f_1(\bm{x},\bm{x}_{\tau},\bm{d}),f_2(\bm{x},\bm{x}_{\tau},\bm{d}))^T=(y,-0.2{x}_{\tau}+2.0y-0.2x^2y+d)^T,\]
$\mathcal{D}=[-0.01,0.01]$, $\mathcal{X}=[0.5,5]\times [-1.5,3.5]$, $\mathcal{I}_0=[0.9,1.1]\times [0.9,1.1]$ with $\partial \mathcal{I}_0=\cup_{i=1}^4 \mathcal{I}_{0,i}$, where $\mathcal{I}_{0,1}=[0.9,0.9]\times [0.9,1.1]$, 
$\mathcal{I}_{0,2}=[1.1,1.1]\times [0.9,1.1]$, $\mathcal{I}_{0,3}=[0.9,1.1]\times [0.9,0.9]$ and $\mathcal{I}_{0,4}=[0.9,1.1]\times [1.1,1.1]$.
\end{example}

In this example $M'=12.0, M=12.0, N=0.2$, $R=2$ and $\epsilon=2$.  Through simple calculations, $\tau=0.02$ satisfies Theorem \ref{tau}. $K$ is assigned to $250$ and thus the entire time interval is $[0,5.0]$. If using the condition in \cite{XueMFCLZ17} to estimate $\tau$ with  $\|\frac{\partial \bm{g}(\bm{x},\bm{d})}{\partial \bm{x}}\|_{\max}=7.2$,
$\|\frac{\partial \bm{f}(\bm{x},\bm{x}_{\tau},\bm{d})}{\partial \bm{x}}\|_{\max}=7$, $\| \frac{\partial\bm{f}(\bm{x},\bm{x}_{\tau},\bm{d})}{\partial \bm{x}_{\tau}}\|_{\max}=0.2$, $R=2$ and $\epsilon=2$, we have $\tau \leq 0.008$.

\begin{figure}[h!]
\centering
\scalebox{0.5}{\centerline{\psfig{figure=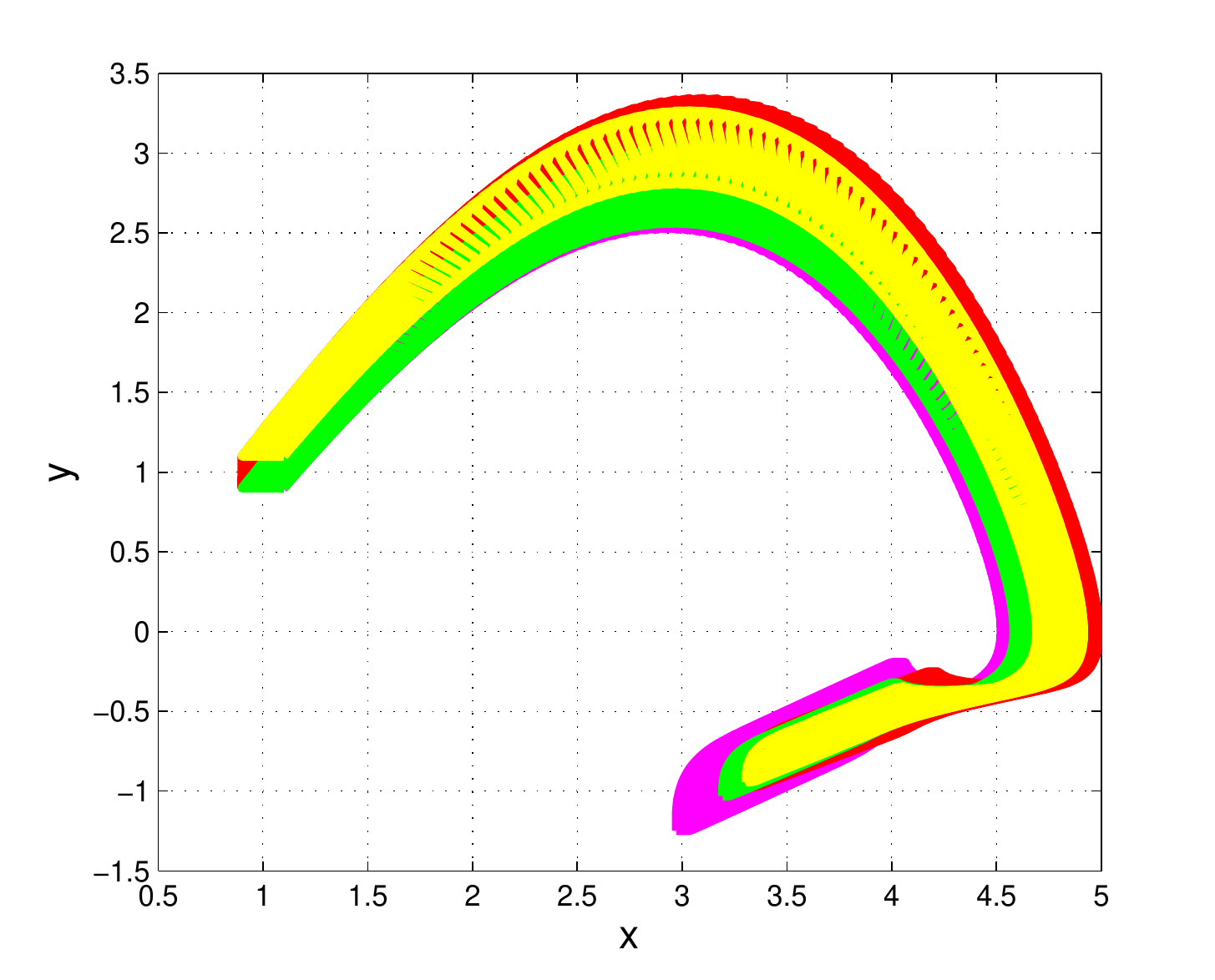,width=15cm,height=15cm}}}
\caption{Reach sets of the set $\partial \mathcal{I}_0$ for Example \ref{lotka-volterra}. Red, purple, green and yellow regions denote $\cup_{t\in [0,5.0]}O(t;\mathcal{I}_{0,1})$, $\cup_{t\in [0,5.0]} O(t;\mathcal{I}_{0,2})$, $\cup_{t\in [0,5.0]} O(t;\mathcal{I}_{0,3})$ and $\cup_{t\in [0,5.0]} O(t;\mathcal{I}_{0,4})$, respectively.}
\label{fig:1}
\end{figure}

\begin{figure}[h!]
\centering
 \includegraphics[width=3.5in,height=3.0in]{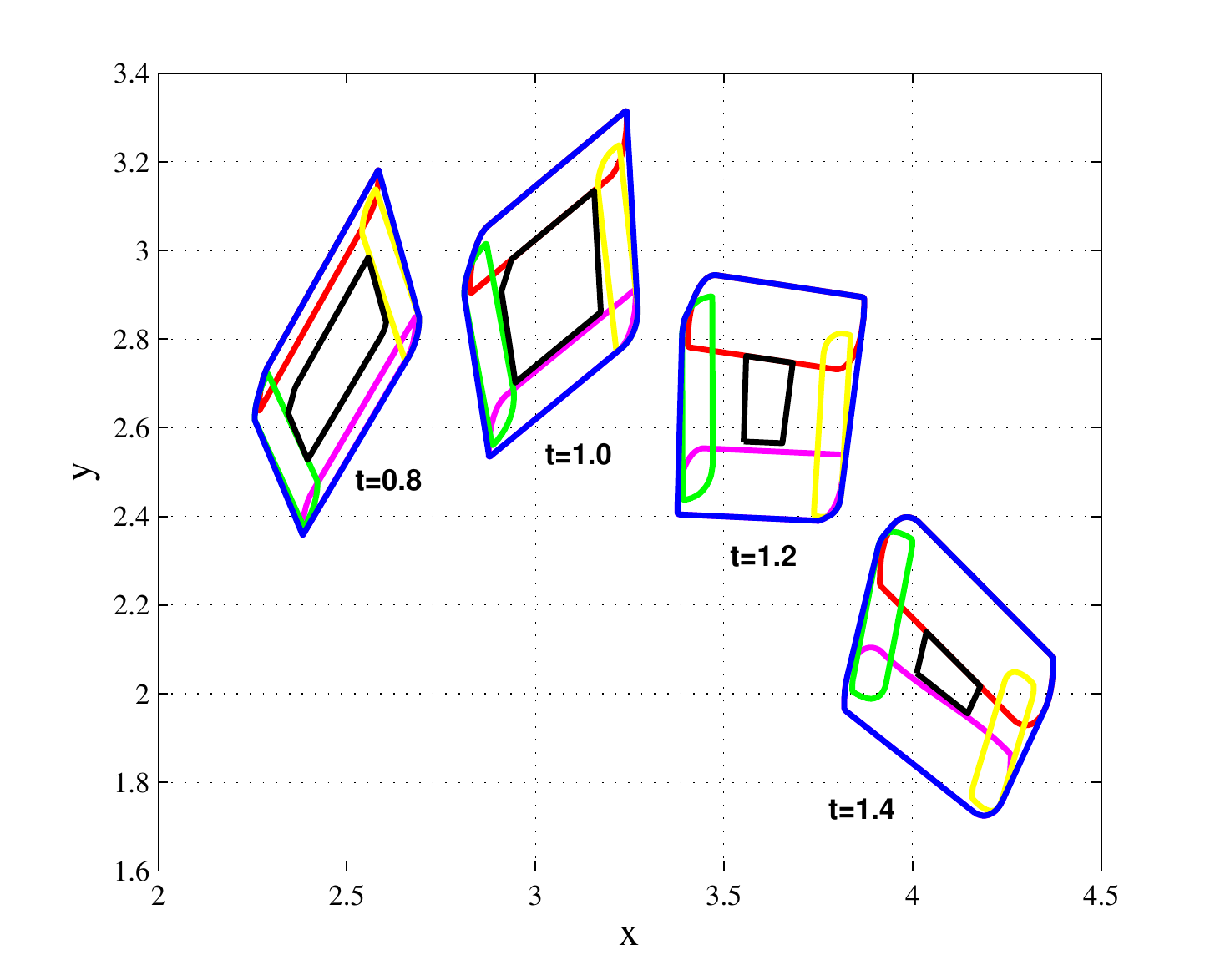}
 \hfill
 \includegraphics[width=3.5in,height=3.0in]{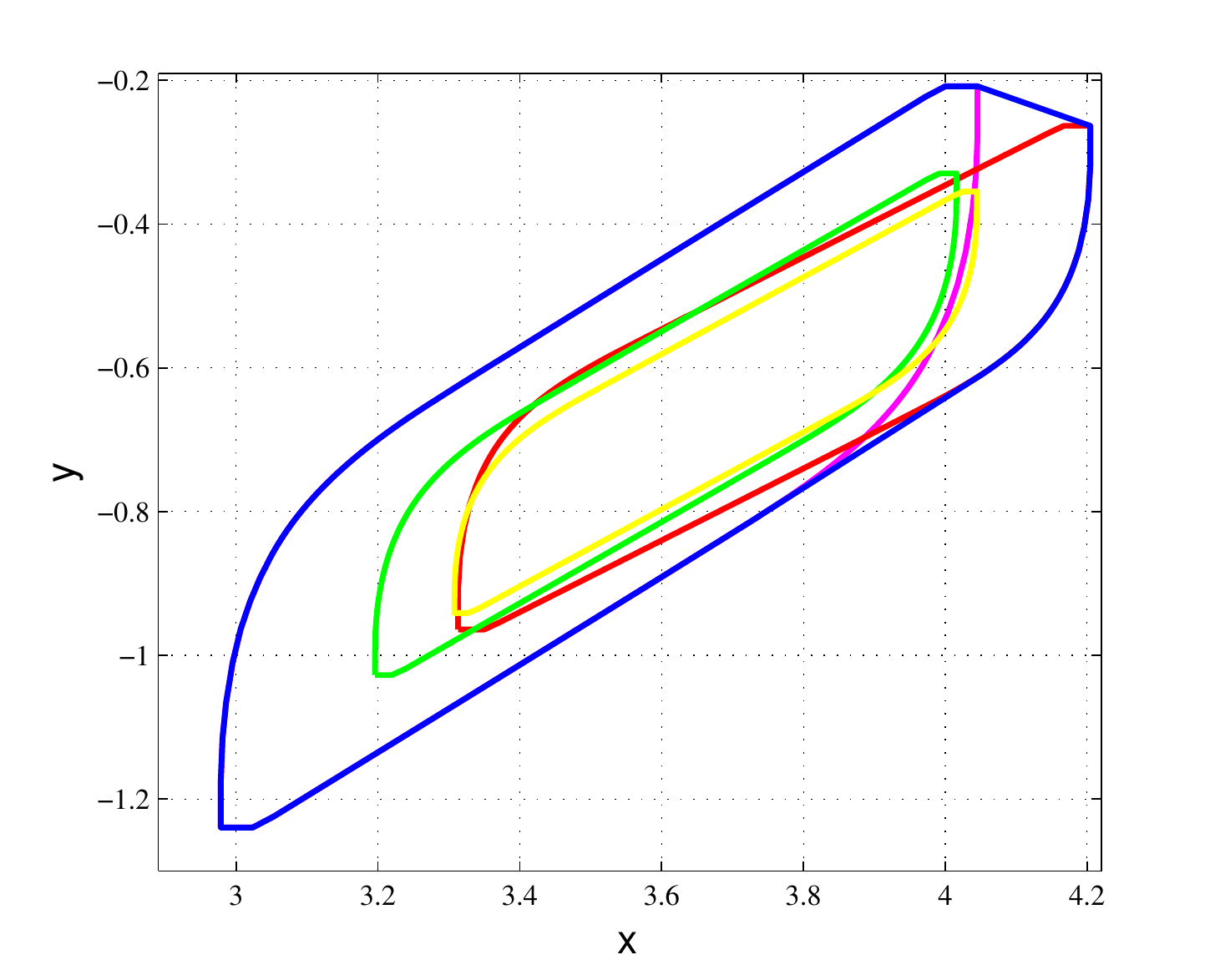}
\caption{Left (Reach sets for Example \ref{lotka-volterra} when $t=0.8, 1.0, 1.2,1.4$): Red, purple, green and yellow curves denote $\partial O(t;\mathcal{I}_{0,1})$, $\partial O(t;\mathcal{I}_{0,2})$, $\partial O(t;\mathcal{I}_{0,3})$ and $\partial O(t;\mathcal{I}_{0,4})$ respectively. Blue curve denotes $\partial O(t;\mathcal{I}_0)$. Black curve denotes $\partial U(t;\mathcal{I}_0)$. Right (Reach sets of the initial set's boundary for Example \ref{lotka-volterra} at time $t=5.0$): Red, purple, green and blue curves denote $\partial O(t;\mathcal{I}_{0,1})$, $\partial O(t;\mathcal{I}_{0,2})$, $\partial O(t;\mathcal{I}_{0,3})$ and $\partial O(t;\mathcal{I}_{0,4})$, respectively. Blue curve denotes $\partial O(t;\mathcal{I}_0)$.}
\label{fig:4}
\end{figure}

We observe that the computed over-approximation $\cup_{t\in [0,5.0]} O(t;\partial \mathcal{I}_0)$ is included in $\mathcal{X}$ as illustrated in Fig. \ref{fig:1}. Unlike that in step 1) in Subsection \ref{rsc}, $\cup_{t\in [0,5.0]} O(t;\partial \mathcal{I}_0)$ illustrated in Fig. \ref{fig:1} is equal to $\cup_{i=1}^4\cup_{t\in [0,5.0]} O(t;\mathcal{I}_{0,i})$ and is computed without the assumption that this system starting from $\mathcal{I}_0$ evolves in the viable domain $\mathcal{X}$ within the time interval $[0,5.0]$. According to Lemma 1 in \cite{IEEE16}, $\cup_{t\in [0,5.0]}\Omega_1(t;\mathcal{I}_0)\subset \mathcal{X}$ in Assumption \ref{assump} is guaranteed. 

An illustration of the computed over- and under-approximations at time $t=0.8,1.0,1.2,1.4$ is demonstrated in Fig. \ref{fig:4}, which also shows the corresponding over-approximation of the reach set of the initial set's boundary. From Fig. \ref{fig:4}, we observe that the computed under-approximation tends to be empty with the time horizon expanding. This contrasts with the wrapping effect in over-approximating the reach set. The computed under-approximation at time $t=1.4$ becomes small as shown in Fig. \ref{fig:4}. We did not yield an under-approximation at time $t=5.0$, thereby only an over-approximation at time $t=5.0$ is showcased in Fig. \ref{fig:4}. The computation time for this reachability analysis is $320.56$ seconds. Partitioning the initial set's boundary into small subsets and performing reachability analysis on each subset independently would help reduce the wrapping effect in over-approximating the reach set of the initial set's boundary, thereby having the possibility to construct a non-empty under-approximation of the set $\Omega_2(5.0;\mathcal{I}_0)$ as well as a more accurate over-approximation of the set $\Omega_1(5.0;\mathcal{I}_0)$.

\begin{example}
\label{seven}
Consider a seven-dimensional system adapted from \cite{XueMFCLZ17}, where $\bm{g}(\bm{x},\bm{d})=\bm{0}$, $\bm{f}(\bm{x},\bm{x}_{\tau},\bm{d})=(1.4x_3-0.9x_{1,\tau},2.5x_5-1.5x_2,0.6x_7-0.8x_3x_2, d-1.3x_4x_3,0.7x_1-1.0x_4
x_5,0.3x_1-3.1x_6,1.8x_6-1.5x_7x_2)^T$, $\mathcal{I}_0=[1.1,1.3]\times [0.95,1.15]\times [1.4,1.6]\times [2.3, 2.5]\times [0.9,1.1]\times [0.0, 0.2]\times [0.35,0.55]$, $\mathcal{X}=[0.5,1.5]\times [0.5,1.5]\times [1.0,2.0]\times [2.0,3.0]\times [0.5,1.5]\times [0.0,0.5]\times [0.0,1.0]$ and $\mathcal{D}=[1.9,2.1]$.
\end{example}

Via simple calculations, we obtain $M'=0,M=6.5, N=0.9$, $R=2$ and $\epsilon=2$ and thus $\tau\leq 0.03$ satisfies the requirement in Theorem \ref{tau}. Also, $\tau$ and $K$ are assigned to $0.02$ and $5$ respectively and thus the entire time interval is $[0,0.1]$. If using the condition in \cite{XueMFCLZ17} with $\|\frac{\partial \bm{g}(\bm{x},\bm{d})}{\partial \bm{x}}\|_{\max}=0$,
$\|\frac{\partial \bm{f}(\bm{x},\bm{x}_{\tau},\bm{d})}{\partial \bm{x}}\|_{\max}=3.9$, $\| \frac{\partial\bm{f}(\bm{x},\bm{x}_{\tau},\bm{d})}{\partial \bm{x}_{\tau}}\|_{\max}=0.9$, $R=2$ and $\epsilon=2$, we have $\tau\leq 8.95\times 10^{-4}$.

We use the same technique as in Example \ref{lotka-volterra} to verify the assumption that $\cup_{t\in [0,0.1]}\Omega_1(t;\mathcal{I}_0)\subset \mathcal{X}$. The computed over-approximation of $\Omega_1(0.1;\mathcal{I}_0)$ is
\begin{equation*}
\begin{split}
 O(0.1;\mathcal{I}_0)=&[1.1556,1.3955]\times [1.0016,1.2165]\times [1.3106,1.5286]\\
&\times [2.0898,2.3308]\times
[0.7957,0.9946]\times [0.02403,0.1847]\times[0.3017,0.5164]
\end{split}
\end{equation*}
 and the computed under-approximation of $\Omega_2(0.1;\mathcal{I}_0)$  is 
 \begin{equation*}
 \begin{split}
U(0.1;\mathcal{I}_0)=&[1.2242,1.3268]\times [1.0703,1.1478]\times [1.3792,1.4600]\\
&\times [2.1585,2.2621]\times
[0.8644,0.9259]\times [0.0927,0.1161]\times [0.3704,0.4478].
\end{split}
\end{equation*}

For this system having variables $\bm{x}$ of seven dimension, we compute an interval over-approximation $O(0.1;\mathcal{I}_0)$ and under-approximation $U(0.1;\mathcal{I}_0)$ respectively. As opposed to polytopic representations, the interval representation may be simpler and more conservative, but the overall computation time consumed is admissible, which is just $130.21$
seconds. 

Based on Examples \ref{11}$\sim$\ref{seven} we found that condition \eqref{uppp} indeed helps to achieve a reasonable improvement on the bound of the term $\tau$ over the condition in \cite{XueMFCLZ17}, by comparing the bounds obtained by condition \eqref{uppp} and the condition in \cite{XueMFCLZ17}. We should point out here that the constraints $\tau\leq 0.66$, $\tau\leq 0.008$ and $\tau\leq 8.95\times 10^{-4}$ for Examples \ref{11}$\sim$\ref{seven} respectively are obtained based on the condition in the revised version of \cite{XueMFCLZ17}, which can be downloaded from \url{http://lcs.ios.ac.cn/~xuebai/Publications.html}. The underlying reason for this improvement is that the derivation of condition \eqref{uppp} only involves operations of the infinity norm of matrices, as reflected in the proofs of Lemma \ref{ode} and \ref{findt}. However, the derivation in \cite{XueMFCLZ17} involves manipulating 2-norm, infinity norm and max norm of matrices and their non-equivalent interconvertibility, thereby introducing conservativeness. Please refer to Lemma 2 and 3 in \cite{XueMFCLZ17} for details. Furthermore, some potential reasons, why the upper bound obtained from condition \eqref{uppp} is still small for certain cases, are presented here: 1). The exact upper bound is not known and thus the conservativeness of condition \eqref{uppp} cannot be evaluated. 2). The way to estimate the time-lag term such that the sensitivity matrices are strictly diagonally dominant is not the best. Actually, it is enough to estimate it such that the sensitivity matrices are invertible, as reflected in Corollary \ref{coro1}. However, how to derive such a bound, which results in invertible sensitivity matrices which are not necessarily strictly diagonally dominant, is still unclear.

\section{Conclusion}
\label{con}
In this note we extended the set-boundary reachability analysis method for perturbation-free ODEs to a class of DDEs subject to time-varying Lipschitz continuous perturbations. Three illustrative examples were employed to demonstrate the performance of our method.

\bibliographystyle{abbrv}
\bibliography{reference}  

\begin{thebibliography}{10}

\bibitem{Althoff13}
M.~Althoff.
\newblock Reachability analysis of nonlinear systems using conservative
  polynomialization and non-convex sets.
\newblock In {\em HSCC'13}, pages 173--182, 2013.

\bibitem{Althoff08}
M.~Althoff, O.~Stursberg, and M.~Buss.
\newblock Reachability analysis of nonlinear systems with uncertain parameters
  using conservative linearization.
\newblock In {\em CDC'08}, pages 4042--4048, 2008.

\bibitem{Asarin01}
E.~Asarin, T.~Dang, and O.~Maler.
\newblock d/dt: A tool for reachability analysis of continuous and hybrid
  systems.
\newblock In {\em NOLCOS'01}, 2001.

\bibitem{Benvenuti08}
L.~Benvenuti, D.~Bresolin, A.~Casagrande, P.~Collins, A.~Ferrari, E.~Mazzi,
  A.~Sangiovanni-Vincentelli, and T.~Villa.
\newblock Reachability computation for hybrid systems with ariadne.
\newblock {\em IFAC Proceedings Volumes}, 41(2):8960--8965, 2008.

\bibitem{Xin12}
X.~Chen, E.~{\'{A}}brah{\'{a}}m, and S.~Sankaranarayanan.
\newblock Taylor model flowpipe construction for non-linear hybrid systems.
\newblock In {\em RTSS'12}, pages 183--192, 2012.

\bibitem{feng2019taming}
S.~Feng, M.~Chen, N.~Zhan, M.~Fr{\"a}nzle, and B.~Xue.
\newblock Taming delays in dynamical systems.
\newblock In {\em CAV'19}, pages 650--669, 2019.

\bibitem{Girard08}
A.~Girard and C.~L. Guernic.
\newblock Zonotope/hyperplane intersection for hybrid systems reachability
  analysis.
\newblock In {\em {HSCC}'08}, pages 215--228, 2008.

\bibitem{GoubaultP17}
E.~Goubault and S.~Putot.
\newblock Forward inner-approximated reachability of non-linear continuous
  systems.
\newblock In {\em {HSCC}'17}, pages 1--10, 2017.

\bibitem{Goubault18}
E.~Goubault, S.~Putot, and L.~Sahlman.
\newblock Inner and outer approximating flowpipes for delay differential
  equations.
\newblock In {\em {CAV}'18}, pages 61--70, 2018.

\bibitem{Huang16}
Z.~Huang, C.~Fan, and S.~Mitra.
\newblock Bounded invariant verification for time-delayed nonlinear networked
  dynamical systems.
\newblock {\em Nonlinear Analysis: Hybrid Systems}, 23:211--229, 2017.

\bibitem{Mitchell07}
I.~M. Mitchell.
\newblock Comparing forward and backward reachability as tools for safety
  analysis.
\newblock In {\em HSCC'07}, pages 428--443, 2007.

\bibitem{Mitchell05}
I.~M. Mitchell, A.~M. Bayen, and C.~J. Tomlin.
\newblock A time-dependent hamilton-jacobi formulation of reachable sets for
  continuous dynamic games.
\newblock {\em IEEE Trans. Automat. Contr.}, 50(7):947--957, 2005.

\bibitem{Alur04}
S.~Prajna and A.~Jadbabaie.
\newblock Safety verification of hybrid systems using barrier certificates.
\newblock In {\em HSCC'04}, pages 477--492, 2004.

\bibitem{Prajna05}
S.~Prajna and A.~Jadbabaie.
\newblock Methods for safety verification of time-delay systems.
\newblock In {\em CDC'05}, pages 4348--4353. IEEE, 2005.

\bibitem{rudin1964}
W.~Rudin et~al.
\newblock {\em Principles of mathematical analysis}, volume~3.
\newblock McGraw-hill New York, 1964.

\bibitem{Varah75}
J.~M. Varah.
\newblock A lower bound for the smallest singular value of a matrix.
\newblock {\em Linear Algebra and its Applications}, 11(1):3--5, 1975.

\bibitem{weaver1999}
N.~Weaver.
\newblock {\em Lipschitz algebras}.
\newblock World Scientific, 1999.

\bibitem{IEEE16}
B.~Xue, A.~Easwaran, N.-J. Cho, and M.~Fr\"anzle.
\newblock Reach-avoid verification for nonlinear systems based on boundary
  analysis.
\newblock {\em IEEE Trans. Automat. Contr.}, pages 3518--3523, 2016.

\bibitem{xue2017just}
B.~Xue, M.~Fr{\"a}nzle, and P.~N. Mosaad.
\newblock Just scratching the surface: Partial exploration of initial values in
  reach-set computation.
\newblock In {\em CDC'17}, pages 1769--1775. IEEE, 2017.

\bibitem{xue2020}
B.~Xue, M.~Fr{\"a}nzle, and N.~Zhan.
\newblock Inner-approximating reachable sets for polynomial systems with
  time-varying uncertainties.
\newblock {\em To Appear in IEEE Trans. on Automat. Contr.}, 2020. DOI:
  10.1109/TAC.2019.2923049.

\bibitem{XueMFCLZ17}
B.~Xue, P.~N. Mosaad, M.~Fr{\"{a}}nzle, M.~Chen, Y.~Li, and N.~Zhan.
\newblock Safe over- and under-approximation of reachable sets for delay
  differential equations.
\newblock In {\em {FORMATS}'17}, pages 281--299, 2017. The revised version can
  be downloaded from \url{http://lcs.ios.ac.cn/~xuebai/Publications.html}.

\bibitem{Xue16}
B.~Xue, Z.~She, and A.~Easwaran.
\newblock Under-approximating backward reachable sets by polytopes.
\newblock In {\em {CAV}'16}, pages 457--476, 2016.

\bibitem{Zou15}
L.~Zou, M.~Fr{\"{a}}nzle, N.~Zhan, and P.~N. Mosaad.
\newblock Automatic verification of stability and safety for delay differential
  equations.
\newblock In {\em {CAV}'15}, pages 338--355, 2015.

\end{thebibliography}
\section{Appendix}
\label{appendix}
In Appendix we show the derivation of Theorem \ref{tau}. Its derivation is based on the requirement that the sensitivity matrices are strictly diagonally dominant.

\subsection{Sensitivity Matrices}
\label{SAT}
In this subsection we introduce sensitivity matrices and associated sensitivity equations.

When $t\in[0,\tau]$, system \eqref{dde5} is governed by ODE $\dot{\bm{x}}(t)=\bm{g}\big(\bm{x}(t),\bm{d}(t)\big)$ with $\bm{d}(t)\in \mathcal{D},$ its flow mapping $\bm{\phi}(t;\bm{x}_0,\bm{d})$ as a function of $\bm{x}_0$ is differentiable with respect to the initial state $\bm{x}_0$. Given a perturbation input $\bm{d}\in \mathcal{D}$, the \textit{sensitivity matrix} of solutions at time $t \in [0,\tau]$ with respect to initial conditions is $
s_{\bm{x}_0}^{\bm{d}}(t)=\frac{\partial \bm{\phi}(t;\bm{x}_0,\bm{d})}{\partial \bm{x}_0},$
which is the solution to ODE: 
\begin{equation}
\label{ode_se}
\dot{s}_{\bm{x}_0}^{\bm{d}}=D_{\bm{g}}s_{\bm{x}_0}^{\bm{d}}, s_{\bm{x}_0}^{\bm{d}}(0)=\bm{I},
\end{equation}
where $s_{\bm{x}_0}^{\bm{d}}(t)\in \mathbb{R}^{n\times n}$, $D_{\bm{g}}$ is the Jacobian matrix of vector field $\bm{g}$ along the trajectory $\bm{\phi}(t;\bm{x}_0,\bm{d})$, i.e. $D_{\bm{g}}=\frac{\partial \bm{g}}{\partial \bm{x}}\mid_{\bm{x}=\bm{\phi}(t;\bm{x}_0,\bm{d})}$, and $\bm{I}\in \mathbb{R}^{n\times n}$ is the identity matrix.

Assume that the solution mapping $\bm{\phi}(t;\bm{x}_0,\bm{d})$ of system \eqref{dde5} for 
$t\in [(k-1)\tau, k\tau]$ and $\bm{x}_0\in \mathcal{I}_0$ could be equivalently reformulated as a continuously differentiable function $\bm{\psi}_{k-1}(t;\bm{x}((k-1)\tau),(k-1)\tau)$ of the state variable $\bm{x}((k-1)\tau)\in \Omega((k-1)\tau;\mathcal{I}_0,\bm{d})$ and the time variable $t\in [(k-1)\tau,k\tau]$, where $k\in \{1,\ldots,K-1\}$ and $\bm{x}((k-1)\tau)=\bm{\phi}((k-1)\tau;\bm{x}_0,\bm{d})$. Also, assume the determinant of the Jacobian matrix of the mapping $\bm{\psi}_{k-1}(t;\bm{x}((k-1)\tau),(k-1)\tau)$ with respect to state $\bm{x}((k-1)\tau)\in \Omega((k-1)\tau;\mathcal{I}_0,\bm{d})$ is not zero for $t\in [(k-1)\tau,k\tau]$. Then, we deduce what follows.
\begin{lemma}
\label{sensi}
Given the above assumptions and a perturbation input $\bm{d}\in \mathcal{D}$, $s_{\bm{x}(k\tau)}^{\bm{d}}(t)=\frac{\partial \bm{x}(t)}{\partial \bm{x}({k\tau})}$ for system \eqref{dde5} exists and satisfies the following equation:
\begin{equation}
\label{sae}
\begin{split}
\dot{s}_{\bm{x}(k\tau)}^{\bm{d}}(t)=&\frac{\partial \bm{f}(\bm{x}(t),\bm{x}_{\tau}(t),\bm{d}(t))}{\partial \bm{x}(t)}s_{\bm{x}(k\tau)}^{\bm{d}}(t)+\frac{\partial \bm{f}(\bm{x}(t),\bm{x}_{\tau}(t),\bm{d}(t))}{\partial \bm{x}_{\tau}(t)}\frac{\partial \bm{x}_{\tau}(t)}{\partial \bm{x}(k\tau)},
\end{split}
\end{equation}
where $\dot{s}_{\bm{x}(k\tau)}^{\bm{d}}=\frac{d~s_{\bm{x}(k\tau)}^{\bm{d}}(t)}{d~t}$, $t\in [k\tau,(k+1)\tau]$, $s_{\bm{x}(k\tau)}^{\bm{d}}(k\tau)=\bm{I}\in \mathbb{R}^{n\times n}$ and $\bm{x}(k\tau)\in \Omega(k\tau;\mathcal{I}_0,\bm{d})$.
\end{lemma}
The proof of Lemma \ref{sensi} is analogous to Lemma 1 in \cite{XueMFCLZ17}.

\begin{corollary}
\label{coro1}
For $\bm{d}\in \mathcal{D}$ and $t\in [k\tau,(k+1)\tau]$, if the determinant of $s_{\bm{x}(k\tau)}^{\bm{d}}(t)$ with respect to $\bm{x}({k\tau})\in \Omega(k\tau;\mathcal{I}_0,\bm{d})$ at time $k \tau $ is not zero, where $k\in\{0,1,\ldots,K-1\}$, then $\bm{x}(t)=\bm{\phi}(t;\bm{x}_0,\bm{d})$ is uniquely determined by $\bm{x}(k\tau)$. 
\end{corollary}

\subsection{Bounding Time-Lag Terms}
\label{GBCT}
In this subsection we derive condition \eqref{uppp} on $\tau$ in system \eqref{dde5}, which renders the sensitivity matrices introduced in Subsection \ref{SAT} strictly diagonally dominant. We begin with the time interval $[0,\tau]$. In the rest we denote $|A_{ii}|-\sum_{j\neq i}|A_{ij}|$ by $\Delta_i(A), 1\leq i\leq n,$ where $A\in \mathbb{R}^{n\times n}$ is a matrix and $A_{ij}$ is the entry in the $i_{th}$ row and $j_{th}$ column of $A$.

\begin{lemma}
\label{ode}
There exist $R>1$ and $\epsilon>1$ such that if \[\tau\leq \min\{\frac{\epsilon-1}{\epsilon M'R},\frac{R-1}{M'R}\},\] the sensitivity matrix $s_{\bm{x}_0}^{\bm{d}}(t)$ in Eq. \eqref{ode_se} is strictly diagonally dominant. Moreover, $\|s_{\bm{x}_0}^{\bm{d}}(t)\|_{\infty}\leq R$ and $\max_{1\leq i\leq n}\frac{1}{\Delta_{i}(s_{\bm{x}_0}^{\bm{d}}(t))} \leq \epsilon$ for $t\in [0,\tau]$, $\bm{d}\in \mathcal{D}$ and $\bm{x}_0\in \mathcal{I}_0$, where $M'$ is defined in Assumption \ref{assump}.
\end{lemma}
\begin{proof}
\eqref{ode_se} tells that the diagonal element in the $i_{th}$ row of $s^{\bm{d}}_{\bm{x}_0}(t)$ for $t\in [0,\tau]$ equals
$
1+\int_{0}^{t}\frac{\partial g_i(\bm{x}(s),\bm{d}(s))}{\partial \bm{x}(s)} \frac{\partial \bm{x}(s)}{\partial x_{0,i}}ds,$
and the element in the $i_{th}$ row and $j_{th}$ column equals
$\int_{0}^{t}\frac{\partial g_i(\bm{x}(s),\bm{d}(s))}{\partial \bm{x}(s)} \frac{\partial \bm{x}(s)}{\partial x_{0,j}}ds,$
where $j\in \{1,\ldots,n\}\setminus \{i\}$.
Therefore,
\[\|s^{\bm{d}}_{\bm{x}_0}(t)\|_{\infty}\leq 1+\tau [\|\frac{\partial \bm{g}(\bm{x},\bm{d})}{\partial \bm{x}}\|_{\infty}\cdot \|\frac{\partial \bm{x}}{\partial \bm{x}_{0}}\|_{\infty}=1+\tau[\|\frac{\partial \bm{g}(\bm{x},\bm{d})}{\partial \bm{x}}\|_{\infty}\cdot \|s^{\bm{d}}_{\bm{x}_0}(t)\|_{\infty}].\]
Since $\|\frac{\partial \bm{g}(\bm{x},\bm{d})}{\partial \bm{x}}\|_{\infty}\leq M'$ for $t\in [0,\tau]$ and $\bm{d}\in \mathcal{D}$, we have
\[\|s^{\bm{d}}_{\bm{x}_0}(t)\|_{\infty}\leq 1+\tau M' \|s^{\bm{d}}_{\bm{x}_0}(t)\|_{\infty},\] implying that
$\|s^{\bm{d}}_{\bm{x}_0}(t)\|_{\infty}\leq \frac{1}{1-\tau M'}$ if $1-\tau M'>0.$ Thus, $\|s^{\bm{d}}_{\bm{x}_0}(t)\|_{\infty}\leq R$ for $t\in [0,\tau]$ and $\bm{d}\in \mathcal{D}$ if
$\tau\leq \frac{R-1}{M'R}.$

Since $1-\int_{0}^{t}\sum_{j=1}^n|\frac{\partial g_i(\bm{x}(s),\bm{d}(s))}{\partial \bm{x}(s)} \frac{\partial \bm{x}(s)}{\partial x_{0,j}}|ds$ and $\|s^{\bm{d}}_{\bm{x}_0}(t)\|_{\infty}\leq R$, we have $\Delta_i(s^{\bm{d}}_{\bm{x}_0}(t))$ is larger than $1-M'R\tau$, where $\sum_{j=1}^n|\frac{\partial g_i(\bm{x},\bm{d})}{\partial \bm{x}} \frac{\partial \bm{x}}{\partial x_{0,j}}|\leq M'R$ can be inferred in the following way:
\[\sum_{j=1}^n|\frac{\partial g_i(\bm{x}(s),\bm{d}(s))}{\partial \bm{x}(s)} \frac{\partial \bm{x}(s)}{\partial x_{0,j}}|\leq \sum_{l=1}^n[|\frac{\partial g_i(\bm{x}(s),\bm{d}(s))}{\partial x_l(s)}|\sum_{j=1}^n |\frac{\partial x_l(s)}{\partial x_{0,j}}|]\leq R\sum_{l=1}^n |\frac{\partial g_i(\bm{x}(s),\bm{d}(s))}{\partial x_l(s)}|\leq M'R.\] $\frac{1}{1-M'R\tau} \leq \epsilon$ implies $\tau\leq \frac{\epsilon-1}{\epsilon M'R}$. Therefore, if \[\tau\leq \min\{\frac{\epsilon-1}{\epsilon M'R},\frac{R-1}{M'R}\},\] $\|s^{\bm{d}}_{\bm{x}_0}(t)\|_{\infty}\leq R$ and $\max_{1\leq i \leq n}\frac{1}{\Delta_i(s^{\bm{d}}_{\bm{x}_0}(t))}\leq \epsilon$ hold. Also, since 
$\tau\leq \frac{\epsilon-1}{\epsilon M'R}$, $1-M'R\tau >0$ holds and thus $s^{\bm{d}}_{\bm{x}_0}(t)$ is strictly diagonally dominant for $t\in [0,\tau]$ and $\bm{d}\in \mathcal{D}$.\qed
\end{proof}

Assume that $s_{\bm{x}((k-1)\tau)}^{\bm{d}}(t)$ is strictly diagonally dominant such that $\|s_{\bm{x}((k-1)\tau)}^{\bm{d}}(t)\|_{\infty} \leq R$ and $\max_{1\leq i\leq n}\frac{1}{\Delta_{i}(s_{\bm{x}((k-1)\tau)}^{\bm{d}}(t))} \leq \epsilon$ for $t\in [(k-1)\tau,k\tau]$ and $\bm{d}\in \mathcal{D}$, where $k\in \{1,\ldots,K-1\}$, $\epsilon >1$ and $R>1$. We construct a constraint on the time-lag term $\tau$, which renders the sensitivity matrix in \eqref{sae} strictly diagonally dominant for $t\in [k\tau,(k+1)\tau]$ as follows.
\begin{lemma}
\label{findt}
Given the above assumption, if the time-lag term $\tau$ satisfies 
\[\tau\leq \min \left\{\frac{\epsilon-1}{\epsilon (MR+NR\epsilon)}, \frac{R-1}{MR+NR\epsilon}\right\},\]
where $M$ and $N$ are presented in Assumption \ref{assump}, then $s_{\bm{x}(k\tau)}^{\bm{d}}(t)$ for $t\in [k\tau,(k+1)\tau]$ and $\bm{d}\in \mathcal{D}$ is strictly diagonally dominant with $\|s_{\bm{x}(k\tau)}^{\bm{d}}(t)\|_{\infty}\leq R$ and $\max_{1\leq i\leq n}\frac{1}{\Delta_{i}(s_{\bm{x}(k\tau)}^{\bm{d}}(t))} \leq \epsilon$.
\end{lemma}
\begin{proof}
Since $s_{\bm{x}((k-1)\tau)}^{\bm{d}}(t)$ is strictly diagonally dominant with $\max_{1\leq i\leq n}\frac{1}{\Delta_{i}(s_{\bm{x}((k-1)\tau)}^{\bm{d}}(t))} \leq \epsilon$ for $t\in [(k-1)\tau,k\tau]$ and $\bm{d}\in \mathcal{D}$,  $\|A^{-1}(t)\|_{\infty}\leq \epsilon$
holds \cite{Varah75}, where $A(t)=s_{\bm{x}((k-1)\tau)}^{\bm{d}}(t)$, $t\in [(k-1)\tau,k\tau]$ and $k\in \{1,\ldots,K-1\}$. 
According to Lemma \ref{sensi}, $s_{\bm{x}(k\tau)}^{\bm{d}}(t)$ for $t\in [k\tau,(k+1)\tau]$ with respect to $\bm{x}(k\tau)$ satisfies
\[s_{\bm{x}(k\tau)}^{\bm{d}}(t)=\bm{I}+\int_{k\tau}^t [\frac{\partial \bm{f}(\bm{x}(s),\bm{x}_{\tau}(s),\bm{d}(s))}{\partial \bm{x}(s)}s_{\bm{x}(k\tau)}^{\bm{d}}(s)+\frac{\partial \bm{f}(\bm{x}(s),\bm{x}_{\tau}(s),\bm{d}(s))}{\partial \bm{x}_{\tau}(s)}\frac{\partial \bm{x}_{\tau}(s)}{\partial \bm{x}(k\tau)}]ds,\]
implying
\[s_{\bm{x}(k\tau)}^{\bm{d}}(t)=\bm{I}+\int_{k\tau}^t[\frac{\partial \bm{f}(\bm{x}(s),\bm{x}_{\tau}(s),\bm{d}(s))}{\partial \bm{x}(s)}s_{\bm{x}(k\tau)}^{\bm{d}}(s)+\frac{\partial \bm{f}(\bm{x}(s),\bm{x}_{\tau}(s),\bm{d}(s))}{\partial \bm{x}_{\tau}(s)}\frac{\partial \bm{x}_{\tau}(s)}{\partial \bm{x}((k-1)\tau)}\frac{\partial \bm{x}((k-1)\tau)}{\partial \bm{x}(k\tau)}]ds.\]

Denote $\bm{x}(k\tau)=(x_{k\tau,1},\ldots,x_{k\tau,n})'$. Thus, the diagonal element in the $i_{th}$ row of the matrix $s^{\bm{d}}_{\bm{x}(k\tau)}(t)$ for $t\in [k\tau, (k+1)\tau]$ is equal to
\[1+\int_{k\tau}^{t}\frac{\partial f_i(\bm{x},\bm{x}_{\tau},\bm{d})}{\partial \bm{x}} \frac{\partial \bm{x}}{\partial x_{k\tau,i}}+\frac{\partial f_i(\bm{x},\bm{x}_{\tau},\bm{d})}{\partial \bm{x}_{\tau}}\frac{\partial \bm{x}_{\tau}}{\partial \bm{x}((k-1)\tau)}\frac{\partial \bm{x}((k-1)\tau)}{\partial x_{k\tau,i}}]ds,\]
the element in the $i_{th}$ row and $j_{th}$ column is equal to
\[
\int_{k\tau}^{t}[\frac{\partial f_i(\bm{x},\bm{x}_{\tau},\bm{d})}{\partial \bm{x}} \frac{\partial \bm{x}}{\partial x_{k\tau,j}}+\frac{\partial f_i(\bm{x},\bm{x}_{\tau},\bm{d})}{\partial \bm{x}_{\tau}}\frac{\partial \bm{x}_{\tau}}{\partial \bm{x}((k-1)\tau)}\frac{\partial \bm{x}((k-1)\tau)}{\partial x_{k\tau,j}}]ds,\]
where $j\in \{1,\ldots,n\}\setminus \{i\}$.
Therefore,
\begin{equation*}
\begin{split}
&\|s^{\bm{d}}_{\bm{x}(k\tau)}(t)\|_{\infty}\leq 1+\tau [\|\frac{\partial \bm{f}(\bm{x},\bm{x}_{\tau},\bm{d})}{\partial \bm{x}}\|_{\infty}\cdot \|\frac{\partial \bm{x}}{\partial \bm{x}_{k\tau}}\|_{\infty}+\|\frac{\partial \bm{f}(\bm{x},\bm{x}_{\tau},\bm{d})}{\partial \bm{x}_{\tau}}\|_{\infty}\cdot \|\frac{\partial \bm{x}_{\tau}}{\partial \bm{x}((k-1)\tau)}\|_{\infty}\cdot \|\frac{\partial \bm{x}((k-1)\tau)}{\partial \bm{x}_{k\tau}}\|_{\infty}]\\
&=1+\tau[\|\frac{\partial \bm{f}(\bm{x},\bm{x}_{\tau},\bm{d})}{\partial \bm{x}}\|_{\infty}\cdot \|s^{\bm{d}}_{\bm{x}(k\tau)}(t)\|_{\infty}+\|\frac{\partial \bm{f}(\bm{x},\bm{x}_{\tau},\bm{d})}{\partial \bm{x}_{\tau}}\|_{\infty}\cdot \|\frac{\partial \bm{x}_{\tau}}{\partial \bm{x}((k-1)\tau)}\|_{\infty}\cdot \|\frac{\partial \bm{x}((k-1)\tau)}{\partial \bm{x}_{k\tau}}\|_{\infty}].
\end{split}
\end{equation*}
Also, since $\|\frac{\partial \bm{f}(\bm{x},\bm{x}_{\tau},\bm{d})}{\partial \bm{x}}\|_{\infty}\leq M$, $\|\frac{\partial \bm{f}(\bm{x},\bm{x}_{\tau},\bm{d})}{\partial \bm{x}_{\tau}}\|_{\infty}\leq N$, $\|\frac{\partial \bm{x}_{\tau}}{\partial \bm{x}((k-1)\tau)}\|_{\infty}\leq R$ and $\|\frac{\partial \bm{x}((k-1)\tau)}{\partial \bm{x}_{k\tau}}\|_{\infty}\leq \epsilon$ for $t\in [k\tau,(k+1)\tau]$ and $\bm{d}\in \mathcal{D}$, we have 
\begin{equation}
\label{eqoo}
s^{\bm{d}}_{\bm{x}(k\tau)}(t)\|_{\infty}\leq 1+\tau M \|s^{\bm{d}}_{\bm{x}(k\tau)}(t)\|_{\infty}+\tau NR\epsilon.
\end{equation}
\eqref{eqoo} implies that $\|s^{\bm{d}}_{\bm{x}(k\tau)}(t)\|_{\infty}\leq \frac{1+\tau NR\epsilon}{1-\tau M}$ if $1-\tau M>0.$ Thus, $\|s^{\bm{d}}_{\bm{x}(k\tau)}(t)\|_{\infty}\leq R$ for $t\in [k\tau,(k+1)\tau]$ and $\bm{d}\in \mathcal{D}$ if
$\tau\leq \frac{R-1}{MR+NR\epsilon}.$

On the other side, since $\|s^{\bm{d}}_{\bm{x}(k\tau)}(t)\|_{\infty}\leq R$, $\Delta_i(s^{\bm{d}}_{\bm{x}(k\tau)}(t))$ is larger than
\[
1-\int_{k\tau}^{t}\sum_{j=1}^n|\frac{\partial f_i(\bm{x}(s),\bm{x}_{\tau}(s),\bm{d}(s))}{\partial \bm{x}(s)} \frac{\partial \bm{x}(s)}{\partial x_{k\tau,j}(s)}+\frac{\partial f_i(\bm{x}(s),\bm{x}_{\tau}(s),\bm{d}(s))}{\partial \bm{x}_{\tau}(s)} \frac{\partial \bm{x}_{\tau}(s)}{\partial \bm{x}((k-1)\tau)} \frac{\partial \bm{x}((k-1)\tau)}{\partial x_{k\tau,j}}|ds,\]
which in turn is larger than $1-(MR+NR\epsilon)\tau$, where the inequality $\sum_{j=1}^n|\frac{\partial f_i(\bm{x}(s),\bm{x}_{\tau}(s),\bm{d}(s))}{\partial \bm{x}(s)} \frac{\partial \bm{x}(s)}{\partial x_{k\tau,j}}|\leq MR$ can be inferred in the following way:
\[
\sum_{j=1}^n|\frac{\partial f_i(\bm{x}(s),\bm{x}_{\tau}(s),\bm{d}(s))}{\partial \bm{x}(s)} \frac{\partial \bm{x}(s)}{\partial x_{k\tau,j}}|\leq \sum_{l=1}^n[|\frac{\partial f_i(\bm{x}(s),\bm{x}_{\tau}(s),\bm{d}(s))}{\partial x_l(s)}|\sum_{j=1}^n |\frac{\partial x_l(s)}{\partial x_{k\tau,j}}|]\leq R\sum_{l=1}^n |\frac{\partial f_i(\bm{x}(s),\bm{x}_{\tau}(s),\bm{d}(s))}{\partial x_l(s)}|\leq MR.\]
Similarly, we obtain
$\sum_{j=1}^n|\frac{\partial f_i(\bm{x}(s),\bm{x}_{\tau}(s),\bm{d}(s))}{\partial \bm{x}_{\tau}(s)} \frac{\partial \bm{x}_{\tau}(s)}{\partial \bm{x}((k-1)\tau)} \frac{\partial \bm{x}((k-1)\tau)}{\partial x_{k\tau,j}}|\leq NR\epsilon.$

$\frac{1}{1-(MR+NR\epsilon)\tau} \leq \epsilon$ implies $\tau\leq \frac{\epsilon-1}{\epsilon(MR+NR\epsilon)}$. Therefore, if
\[\tau\leq \min\{\frac{\epsilon-1}{\epsilon(MR+NR\epsilon)},\frac{R-1}{MR+NR\epsilon}\},\]
$\|s^{\bm{d}}_{\bm{x}(k\tau)}(t)\|_{\infty}\leq R$ and $\max_{1\leq i \leq n}\frac{1}{\Delta_i(s^{\bm{d}}_{\bm{x}(k\tau)}(t))}\leq \epsilon$ hold. Also, $s^{\bm{d}}_{\bm{x}(k\tau)}(t)$ is strictly  diagonally dominant for $t\in [k\tau,(k+1)\tau]$ and $\bm{d}\in \mathcal{D}$ since $\tau\leq \frac{\epsilon-1}{\epsilon(MR+NR\epsilon)}$, $1-(MR+NR\epsilon)\tau >0$ holds.\qed
\end{proof}

In summary, we have Theorem \ref{tau} via combining Lemma \ref{ode} and Lemma \ref{findt}.

\end{document}